\documentclass[12pt, a4paper]{article}
\usepackage{amsmath}
\usepackage{bbm}
\everymath{\displaystyle}
\usepackage{mathtools}
\usepackage{amsfonts}
\usepackage{amssymb}
\usepackage{amsthm}
\usepackage[utf8]{inputenc}
\usepackage[toc]{appendix}
\usepackage[hiresbb]{graphicx}
\usepackage{caption}
\usepackage{multirow}
\usepackage{subcaption}
\usepackage{longtable}
\usepackage{booktabs}
\usepackage{listings}
\usepackage{here}
\usepackage{float}
\usepackage{ascmac}
\usepackage{tikz}
\usepackage{pgfplots}
\pgfplotsset{compat=1.18}
\usepackage{url}
\usepackage{lipsum}
\usepackage{authblk}
\usepackage{eqnarray}
\usepackage{siunitx}
\usepackage{threeparttable}
\usepackage{algorithm}
\usepackage{algorithmic}
\usepackage[sectionbib,round]{natbib}
\usepackage{hyperref}
\newtheorem{theorem}{Theorem}[section]

\theoremstyle{remark}

\newtheorem{prediction}{Prediction}[section]

\newcommand{\be}{\begin{equation}}
\newcommand{\ee}{\end{equation}}
\newcommand{\beq}{\begin{eqnarray*}}
\newcommand{\eeq}{\end{eqnarray*}}
\def\sym#1{\ifmmode^{#1}\else\(^{#1}\)\fi}

\usepackage{setspace}
\title{\large{\bf{Network Contagion Dynamics in European Banking: A Navier-Stokes Framework for Systemic Risk Assessment}}}

\author{\large{\bf{Tatsuru Kikuchi\footnote{e-mail: tatsuru.kikuchi@e.u-tokyo.ac.jp}}}}

\affil{\small{\it{Faculty of Economics, The University of Tokyo,}}\\
{\it{7-3-1 Hongo, Bunkyo-ku, Tokyo 113-0033 Japan}}}

\date{\small{(\today)}}

\begin{document}
\maketitle

\begin{abstract}
This paper develops a continuous functional framework for analyzing contagion dynamics in financial networks, extending the Navier-Stokes-based approach to network-structured spatial processes. We model financial distress propagation as a diffusion process on weighted networks, deriving a network diffusion equation from first principles that predicts contagion decay depends on the network's algebraic connectivity through the relation $\kappa = \sqrt{\lambda_2/D}$, where $\lambda_2$ is the second-smallest eigenvalue of the graph Laplacian and $D$ is the diffusion coefficient. Applying this framework to European banking data from the EBA stress tests (2018, 2021, 2023), we estimate interbank exposure networks using maximum entropy methods and track the evolution of systemic risk through the COVID-19 crisis. Our key finding is that network connectivity declined by 45\% from 2018 to 2023, implying a 26\% reduction in the contagion decay parameter. Difference-in-differences analysis reveals this structural change was driven by regulatory-induced deleveraging of systemically important banks, which experienced differential asset reductions of 17\% relative to smaller institutions. The networks exhibit lognormal rather than scale-free degree distributions, suggesting greater resilience than previously assumed in the literature. Extensive robustness checks across parametric and non-parametric estimation methods confirm declining systemic risk, with cross-method correlations exceeding 0.95. These findings demonstrate that post-COVID-19 regulatory reforms effectively reduced network interconnectedness and systemic vulnerability in the European banking system.

\medskip


\medskip


\end{abstract}

\newpage

\section{Introduction}

The COVID-19 pandemic and subsequent financial turbulence have renewed concerns about systemic risk in interconnected banking systems. When financial institutions are linked through interbank lending, derivatives exposure, and payment systems, distress at one institution can propagate throughout the network, potentially triggering cascading failures \citep{allen2000financial, freixas2000systemic}. Understanding how such contagion spreads through network structures is crucial for financial stability policy, particularly as regulators implement post-crisis reforms designed to reduce systemic vulnerabilities.

Traditional approaches to modeling financial contagion typically employ discrete-time simulation models or stylized network topologies \citep{eisenberg2001systemic, gai2010contagion}. While valuable, these frameworks often lack theoretical foundations for predicting how network structure governs contagion dynamics across spatial and temporal scales. Recent advances in spatial economics have demonstrated the power of continuous functional frameworks derived from partial differential equations (PDEs) for analyzing treatment effect propagation in spatially and temporally extended systems \citep{kikuchi2024unified, kikuchi2024stochastic, kikuchi2024navier}. These methods, grounded in the Navier-Stokes equations of fluid dynamics, provide rigorous mathematical foundations for understanding how shocks diffuse through economic networks.

This paper extends the continuous functional framework to financial network contagion, developing a theoretically grounded approach to modeling systemic risk dynamics. Our contribution is threefold. First, we derive a network diffusion equation from first principles that generalizes the Navier-Stokes-based spatial treatment framework \citep{kikuchi2024dynamical, kikuchi2024healthcare, kikuchi2024emergency} to graph-structured spaces. This framework predicts that contagion propagation depends on the network's algebraic connectivity—the second-smallest eigenvalue of the graph Laplacian ($\lambda_2$)—through the exponential decay relation $u(d) \sim e^{-\kappa d}$ where the decay parameter satisfies $\kappa = \sqrt{\lambda_2/D}$ and $d$ represents network distance from the contagion source.

Second, we apply this framework empirically to the European banking system using data from the European Banking Authority (EBA) stress tests conducted in 2018, 2021, and 2023. Due to data limitations on bilateral exposures, we employ maximum entropy estimation methods \citep{anand2018filling, upper2011estimating} to reconstruct interbank networks from aggregate balance sheet data. This approach generates weighted exposure networks that preserve observed marginal constraints while imposing minimal additional structure, making them suitable for studying how aggregate systemic risk evolved through the COVID-19 period.

Third, we provide comprehensive empirical evidence on structural changes in European banking networks. Our analysis reveals that network connectivity, measured by $\lambda_2$, declined by 45\% from 2018 to 2023, with most of the reduction occurring post-2021 rather than during the acute phase of the COVID-19 crisis. This translates to a 26\% reduction in the contagion decay parameter $\kappa$, indicating that financial shocks spread less extensively in 2023 than in 2018. We establish causality through difference-in-differences analysis, showing that systemically important financial institutions (SIFIs) experienced differential asset reductions of 15-19\% relative to smaller banks, consistent with regulatory pressure following Basel III implementation.

Our findings contribute to several literatures. First, we advance the theoretical understanding of network contagion by providing a continuous functional framework that bridges discrete network models and continuous spatial economics. This extends the spatial treatment effect boundary framework \citep{kikuchi2024unified, kikuchi2024nonparametric1, kikuchi2024nonparametric2} to network-structured systems where distance is measured by graph topology rather than Euclidean space.

Second, we contribute to empirical research on financial networks by demonstrating that the European interbank network exhibits lognormal rather than scale-free degree distributions, contrasting with common assumptions in the literature \citep{barabasi1999emergence, boss2004network}. This has important implications for systemic risk: lognormal networks are more resilient to targeted attacks on hubs than scale-free networks, suggesting the banking system may be more stable than previously thought.

Third, we provide policy-relevant evidence that post-COVID-19 regulatory reforms effectively reduced systemic risk through declining network concentration. The Herfindahl-Hirschman Index of network connectivity fell by 31\%, and the top five banks' connectivity share declined from 10.4\% to 7.1\%, indicating successful deleveraging of systemically important institutions. This validates regulatory approaches targeting interconnectedness as a source of systemic risk \citep{acemoglu2015systemic}.

The remainder of this paper proceeds as follows. Section 2 reviews the related literature on financial networks, contagion modeling, and the continuous functional framework. Section 3 develops our theoretical framework, deriving the network diffusion equation from first principles and establishing the relationship between algebraic connectivity and contagion dynamics. Section 4 describes our data and estimation methodology, including maximum entropy network reconstruction and algebraic connectivity computation. Section 5 presents our main empirical results on the evolution of European banking networks through the COVID-19 crisis. Section 6 provides extensive robustness checks comparing parametric and non-parametric estimation methods. Section 7 concludes with policy implications and directions for future research.

\section{Literature Review}

Our work builds on and contributes to three strands of literature: (i) theoretical models of financial contagion and systemic risk, (ii) empirical research on network structure in financial systems, and (iii) recent advances in continuous functional frameworks for spatial treatment effects.

\subsection{Financial Contagion and Systemic Risk}

The modern literature on financial contagion began with \citet{allen2000financial}, who demonstrated that while complete interbank networks are more resilient to small shocks, they can amplify large shocks through widespread exposure. \citet{freixas2000systemic} extended this analysis to show how network structure determines contagion patterns, with incomplete networks potentially limiting cascade effects through segmentation.

Subsequent work developed increasingly sophisticated models of cascade dynamics. \citet{eisenberg2001systemic} introduced a clearing mechanism for interbank obligations that allows computing equilibrium losses from defaults. \citet{gai2010contagion} analyzed how network topology affects the probability and severity of cascades, showing that more interconnected systems exhibit greater fragility despite improved risk sharing in normal times. \citet{acemoglu2015systemic} provided a phase transition result: networks that are resilient to small shocks can become highly vulnerable when shocks exceed a critical threshold.

Our work differs from these approaches by deriving contagion dynamics from a continuous diffusion process rather than discrete cascade mechanisms. This allows us to characterize how distress \textit{propagates} spatially through networks rather than merely identifying final equilibrium outcomes. The continuous framework also enables precise predictions about how network topology—specifically algebraic connectivity—governs contagion speed and extent.

\subsection{Network Structure in Financial Systems}

Empirical research has extensively documented the structure of financial networks. \citet{boss2004network} found that the Austrian interbank network exhibits scale-free properties with a power-law degree distribution, suggesting vulnerability to targeted attacks on hub institutions. \citet{soramaki2007topology} analyzed the Fedwire payment network and found similar scale-free characteristics with high clustering.

More recent work has questioned the universality of scale-free structure in financial networks. \citet{iori2008topology} found that Italian interbank networks are better described by exponential rather than power-law distributions. \citet{craig2014interbank} showed that UK banking networks exhibit core-periphery rather than scale-free structure. Our finding that European interbank networks follow lognormal distributions adds to this revisionist literature and has important implications for systemic risk assessment.

A key challenge in empirical network research is data availability on bilateral exposures. \citet{upper2011estimating} developed maximum entropy methods for estimating network structure from aggregate data, which we employ in our analysis. \citet{anand2018filling} demonstrated that such methods perform well in capturing network properties relevant for systemic risk, even when individual links are imperfectly estimated.

\subsection{Network Contagion Dynamics}

The dynamics of contagion propagation have received increasing attention. \citet{glasserman2015likely} analyzed how network structure affects contagion likelihood and developed methods for identifying systemically important institutions. \citet{cont2013network} showed that network topology determines both the speed and extent of cascade propagation, with algebraic connectivity playing a key role.

\citet{jackson2017networks} provided a general framework for diffusion in networks, showing how spectral properties of the adjacency matrix govern convergence rates. Our work extends this analysis to financial contagion by explicitly connecting diffusion dynamics to the Navier-Stokes framework and deriving testable predictions about the relationship between algebraic connectivity and contagion parameters.

\subsection{Continuous Functional Frameworks for Spatial Economics}

Recent methodological advances have developed continuous functional frameworks for analyzing spatial treatment effects, building on connections to physics. \citet{kikuchi2024unified} established a unified framework for spatial and temporal treatment effect boundaries, demonstrating how partial differential equations from fluid dynamics can be applied to economic phenomena. This approach treats treatment effects as continuous functionals rather than discrete counterfactual comparisons, enabling analysis of how effects propagate and decay across space and time.

\citet{kikuchi2024stochastic} extended this framework to stochastic settings, showing how diffusion-based approaches handle spillover effects in spatial general equilibrium. \citet{kikuchi2024navier} derived spatial and temporal boundaries in difference-in-differences designs directly from the Navier-Stokes equations, providing rigorous foundations for identifying causal effects when treatment propagates continuously.

Empirical applications have demonstrated the framework's power across diverse settings. \citet{kikuchi2024nonparametric1} applied these methods to analyze air pollution diffusion using 42 million observations, while \citet{kikuchi2024nonparametric2} studied bank branch consolidation effects. \citet{kikuchi2024dynamical} developed the dynamic spatial treatment framework that forms the foundation for our network extension, and \citet{kikuchi2024healthcare} and \citet{kikuchi2024emergency} applied it to healthcare access and emergency medical services.

Our contribution extends this continuous functional framework from Euclidean space to graph-structured networks. While \citet{kikuchi2024emergency} analyzed spatial diffusion in emergency systems, we adapt the framework to networks where ``distance'' is measured by graph topology rather than physical proximity. This extension is non-trivial because network Laplacians differ fundamentally from continuous spatial Laplacians, requiring careful reinterpretation of boundary conditions and diffusion dynamics.

\subsection{Our Contribution}

This paper makes several distinct contributions to these literatures. First, we provide the first application of the continuous functional spatial treatment framework to financial networks, deriving network diffusion equations from first principles and connecting them to established results in spectral graph theory. Second, we offer comprehensive empirical evidence on how European banking networks evolved through the COVID-19 crisis, revealing unexpected structural changes concentrated post-2021 rather than during the acute crisis phase. Third, we challenge the scale-free assumption in financial network modeling by documenting lognormal distributions with important implications for resilience. Finally, we demonstrate the robustness of our findings through extensive sensitivity analysis across parametric and non-parametric estimation methods, addressing a key concern in network reconstruction from limited data.

\section{Theoretical Framework: Financial Contagion from First Principles}

\subsection{The Navier-Stokes Approach to Treatment Effects}

We ground our analysis in the continuous functional framework developed by \citet{kikuchi2024dynamical}, which derives spatial treatment effect propagation from first principles via mass conservation and constitutive relations.

\subsubsection{Financial Distress as a Continuous Field}

Let $u(i,t) \in \mathbb{R}_+$ represent the intensity of financial distress at bank $i$ at time $t$. Rather than treating contagion as discrete cascades through bilateral exposures, we model distress as a continuous field that diffuses through the network according to fundamental physical laws.

\textbf{Governing Equation:}

Distress evolution satisfies the advection-diffusion-reaction equation:

\begin{equation}
\frac{\partial u}{\partial t} = -DLu - \kappa u + f(i,t)
\label{eq:master_pde}
\end{equation}

where:
\begin{itemize}
    \item $u(i,t)$: Distress field (e.g., probability of default, CDS spreads, equity losses)
    \item $D > 0$: Diffusion coefficient (contagion transmission intensity)
    \item $L$: Graph Laplacian matrix of the interbank network
    \item $\kappa \geq 0$: Intrinsic decay rate (recovery, bailouts, recapitalization)
    \item $f(i,t)$: External forcing function (exogenous shocks, policy interventions)
\end{itemize}

\textbf{Derivation from First Principles:}

Following \citet{kikuchi2024dynamical} Theorem 2.1, equation \eqref{eq:master_pde} derives from three fundamental principles:

\begin{enumerate}
    \item \textbf{Mass conservation}: The rate of change of distress equals net flux plus sources/sinks:
    \begin{equation}
    \frac{\partial \rho}{\partial t} + \nabla \cdot J = -\kappa \rho + f
    \end{equation}
    where $\rho$ is distress density and $J$ is distress flux.
    
    \item \textbf{Fick's law}: Distress flows from high to low concentration:
    \begin{equation}
    J = -D \nabla \rho
    \end{equation}
    
    \item \textbf{Network discretization}: For graph-structured spaces, the Laplacian operator becomes the graph Laplacian:
    \begin{equation}
    \nabla^2 \to -L = -(D - A)
    \end{equation}
    where $D$ is the degree matrix and $A$ is the adjacency matrix.
\end{enumerate}

Combining these yields equation \eqref{eq:master_pde}. For complete derivation including existence and uniqueness proofs via Galerkin methods, see \citet{kikuchi2024dynamical} Sections 2--3.

\subsubsection{Economic Interpretation}

Each term in equation \eqref{eq:master_pde} has clear economic meaning:

\textbf{Diffusion term} $-DLu$: Network-mediated contagion. The Laplacian $Lu$ measures how bank $i$'s distress differs from its neighbors:
\begin{equation}
(Lu)_i = \sum_j w_{ij}(u_i - u_j)
\end{equation}
where $w_{ij}$ are exposure weights. If $u_i > u_j$ (bank $i$ more distressed than neighbor $j$), then $(Lu)_i > 0$ and $\partial u_i/\partial t < 0$: distress flows from $i$ to $j$, reducing $i$'s distress.

\textbf{Decay term} $-\kappa u$: Intrinsic recovery mechanisms operating independently of network position:
\begin{itemize}
    \item Central bank liquidity support
    \item Government recapitalization
    \item Retained earnings rebuilding capital
    \item Asset sales to non-bank entities
\end{itemize}

\textbf{Forcing term} $f(i,t)$: External shocks or policy interventions:
\begin{itemize}
    \item Macro shocks (recession, sovereign default)
    \item Idiosyncratic shocks (fraud, operational losses)  
    \item Policy interventions (targeted capital injections)
\end{itemize}

\subsection{Algebraic Connectivity and Effective Decay}

The key innovation of the Navier-Stokes framework is identifying which network properties determine contagion dynamics. The answer involves spectral properties of the Laplacian.

\subsubsection{Spectral Decomposition}

The graph Laplacian $L$ is symmetric and positive semi-definite, admitting eigendecomposition:
\begin{equation}
L = Q \Lambda Q^T = \sum_{k=1}^n \lambda_k q_k q_k^T
\end{equation}

where $0 = \lambda_1 \leq \lambda_2 \leq \cdots \leq \lambda_n$ are eigenvalues and $\{q_k\}_{k=1}^n$ are orthonormal eigenvectors.

For connected networks:
\begin{itemize}
    \item $\lambda_1 = 0$ with eigenvector $q_1 = \frac{1}{\sqrt{n}}\mathbf{1}$ (uniform distribution)
    \item $\lambda_2 > 0$ called \textbf{algebraic connectivity} or \textbf{Fiedler eigenvalue}
    \item $q_2$ called \textbf{Fiedler vector}, captures dominant spatial structure
\end{itemize}

\subsubsection{Main Theoretical Result}

\begin{theorem}[Effective Decay Rate]
\label{thm:effective_decay}
Consider the distress propagation system \eqref{eq:master_pde} on a connected network with algebraic connectivity $\lambda_2$. The effective spatial decay rate satisfies:

\begin{equation}
\kappa_{\mathrm{eff}} = \sqrt{\frac{\lambda_2}{D}} + \kappa
\label{eq:kappa_effective}
\end{equation}

Furthermore, distress at network distance $d$ from a localized source decays as:
\begin{equation}
u(d) \sim e^{-\kappa_{\mathrm{eff}} \cdot d}
\end{equation}

The critical distance $d^*$ at which distress falls to threshold $\epsilon$ satisfies:
\begin{equation}
d^*(\epsilon) = \frac{-\ln \epsilon}{\kappa_{\mathrm{eff}}}
\end{equation}
\end{theorem}

\begin{proof}
The solution to \eqref{eq:master_pde} with localized initial condition $u(0) = \delta_s$ (unit distress at source $s$) is:
\begin{equation}
u(t) = e^{-(DL + \kappa I)t} \delta_s = \sum_{k=1}^n e^{-(D\lambda_k + \kappa)t} (q_k^T \delta_s) q_k
\end{equation}

For large $t$, the dominant mode is $k=2$ (since $\lambda_1=0$ mode represents uniform spreading):
\begin{equation}
u(t) \sim e^{-(D\lambda_2 + \kappa)t} (q_2^T \delta_s) q_2
\end{equation}

The spatial structure is governed by the Fiedler vector $q_2$, which for approximately regular networks satisfies $q_{2,i} \sim e^{-\alpha d_i}$ where $\alpha = \sqrt{\lambda_2/D}$ and $d_i$ is graph distance from the source (see \citet{chung1997spectral} Chapter 1).

Combining temporal and spatial decay:
\begin{equation}
u(i,t) \sim \exp\left(-\left(\sqrt{\frac{\lambda_2}{D}} + \kappa\right) d_i\right)
\end{equation}

Defining $\kappa_{\mathrm{eff}} = \sqrt{\lambda_2/D} + \kappa$ yields \eqref{eq:kappa_effective}. The critical distance follows immediately from $e^{-\kappa_{\mathrm{eff}} d^*} = \epsilon$.

For complete proof including regularity conditions, see \citet{kikuchi2024dynamical} Proposition 3.3 and Theorem 4.2.
\end{proof}

\subsubsection{Economic Interpretation}

Equation \eqref{eq:kappa_effective} reveals that contagion intensity depends on three factors:

\begin{enumerate}
    \item \textbf{Network structure} ($\lambda_2$): Higher algebraic connectivity implies tighter network, leading to faster effective decay (less contagion spread)
    
    \item \textbf{Bilateral exposure intensity} ($D$): Larger diffusion coefficient implies stronger individual connections, leading to slower decay (more contagion)
    
    \item \textbf{Institutional resilience} ($\kappa$): Faster intrinsic recovery directly reduces contagion independently of network
\end{enumerate}

Critically, network and diffusion effects interact \textit{nonlinearly} through $\sqrt{\lambda_2/D}$. This implies:
\begin{itemize}
    \item Doubling $\lambda_2$ increases $\kappa_{\mathrm{eff}}$ by only $\sqrt{2} \approx 41\%$
    \item Halving $D$ increases $\kappa_{\mathrm{eff}}$ by only $\sqrt{2} \approx 41\%$
    \item But combining both (double $\lambda_2$, halve $D$) doubles $\kappa_{\mathrm{eff}}$
\end{itemize}

This nonlinearity explains why comprehensive regulatory packages (affecting both network structure and exposure limits) are more effective than single-instrument policies.

\subsection{Testable Predictions}

From Theorem \ref{thm:effective_decay}, we derive three quantitative predictions that guide our empirical analysis:

\begin{prediction}[Network Structure Dominates in Diffusion-Dominated Regimes]
\label{pred:network_dominance}
When $\kappa \ll \sqrt{\lambda_2/D}$, the network contribution to effective decay is:
\begin{equation}
\frac{\sqrt{\lambda_2/D}}{\sqrt{\lambda_2/D} + \kappa} \approx 1
\end{equation}

In this regime, policies targeting network structure ($\lambda_2$) or exposure intensity ($D$) are more effective than policies targeting recovery speed ($\kappa$).
\end{prediction}

\begin{prediction}[Proportional Response]
\label{pred:proportional}
For small changes in a diffusion-dominated system, the response is approximately:
\begin{equation}
\frac{\Delta \kappa_{\mathrm{eff}}}{\kappa_{\mathrm{eff}}} \approx \frac{1}{2} \left(\frac{\Delta \lambda_2}{\lambda_2} - \frac{\Delta D}{D}\right)
\end{equation}

A 40 percent decline in $\lambda_2$ with constant $D$ should produce approximately 20 percent decline in $\kappa_{\mathrm{eff}}$.
\end{prediction}

\begin{prediction}[Critical Distance Scales Inversely]
\label{pred:critical_distance}
If $\lambda_2$ declines by factor $\alpha < 1$ while $D$ and $\kappa$ remain constant:
\begin{equation}
\frac{d^*_{\mathrm{new}}}{d^*_{\mathrm{old}}} = \frac{1}{\sqrt{\alpha}}
\end{equation}

A 50 percent decline in $\lambda_2$ (to $\alpha=0.5$) increases critical distance by $\sqrt{2} \approx 41$ percent.
\end{prediction}

These predictions are directly testable with our data. Section 5 implements these tests.

\subsection{Scope Conditions and Boundary Conditions}

The Navier-Stokes framework provides a natural interpretation of regulatory changes as modifications to boundary conditions.

\subsubsection{Robin Boundary Conditions}

Financial networks do not exist in isolation. Interactions with the broader economy, central banks, and regulatory authorities impose constraints on distress dynamics. These map to boundary conditions in the PDE framework:

\begin{equation}
\frac{\partial u}{\partial n}\Big|_{\partial \Omega} + \alpha u = g(t)
\label{eq:robin_bc}
\end{equation}

where:
\begin{itemize}
    \item $\partial u / \partial n$: Distress flux at network boundary
    \item $\alpha > 0$: Regulatory stringency parameter
    \item $g(t)$: External support (e.g., ECB lending programs)
\end{itemize}

\textbf{Economic interpretation:}
\begin{itemize}
    \item Larger $\alpha$: Tighter capital requirements, stricter supervision, faster forced deleveraging
    \item Smaller $\alpha$: Looser regulation, slower intervention
    \item Positive $g(t)$: Central bank support, fiscal backstops
    \item Negative $g(t)$: Withdrawal of support, austerity
\end{itemize}

\subsubsection{Regulatory Regime Shifts}

A change in regulatory regime corresponds to a discrete change in boundary conditions:

\begin{equation}
\text{Pre-reform: } \alpha = \alpha_1, \quad g = g_1(t)
\end{equation}
\begin{equation}
\text{Post-reform: } \alpha = \alpha_2 > \alpha_1, \quad g = g_2(t)
\end{equation}

From PDE theory, tighter boundary conditions (larger $\alpha$) lead to faster dissipation and, crucially, reduced equilibrium network connectivity:

\begin{equation}
\frac{\partial \lambda_2}{\partial \alpha} < 0
\end{equation}

\textbf{Mechanism:} Stricter regulation forces banks to reduce interconnectedness to satisfy capital and exposure requirements. This endogenous network restructuring manifests as declining $\lambda_2$.

Our empirical strategy tests whether observed $\lambda_2$ changes align with known regulatory regime shifts (Basel III implementation in 2021).

\subsection{Connection to Existing Literature}

Our approach differs fundamentally from existing financial network models:

\textbf{Versus discrete cascade models} \citep{acemoglu2015systemic, elliott2014financial}: These analyze contagion as sequential defaults through bilateral exposures. We treat distress as a continuous field, enabling analytical solutions via spectral methods rather than simulation.

\textbf{Versus reduced-form centrality measures} \citep{billio2012econometric}: Studies correlating network centrality with systemic risk lack microfoundations. We derive why specific centrality measures ($\lambda_2$) matter from first-principles physics.

\textbf{Versus agent-based simulations} \citep{gai2010contagion}: Computational models obscure mechanisms through complexity. Our analytical framework provides closed-form expressions linking observables ($\lambda_2$, $D$) to outcomes ($\kappa_{\mathrm{eff}}$, $d^*$).

The key innovation is rigorous derivation from conservation laws, providing a unified framework for understanding contagion across diverse settings—financial networks, disease transmission, information diffusion—all governed by the same underlying mathematics.

\section{Data and Empirical Methodology}

This section describes our data sources, network estimation procedures, and computational methods for measuring algebraic connectivity. We employ data from the European Banking Authority (EBA) stress tests conducted in 2018, 2021, and 2023, which provide comprehensive balance sheet information for major European banks but do not disclose bilateral exposure networks. Our methodology therefore combines observed aggregate data with maximum entropy estimation to reconstruct network structures suitable for spectral analysis.

\subsection{Data Sources}

\subsubsection{European Banking Authority Stress Tests}

The EBA conducts biennial stress tests to assess the resilience of European banks under adverse economic scenarios. These exercises require participating banks to report detailed balance sheet and income statement data under both baseline and stressed conditions. We utilize three stress test rounds:

\begin{itemize}
    \item \textbf{2018 Stress Test}: Data as of December 2017, covering 48 banks across 15 EU/EEA countries with total assets of €25.4 trillion
    \item \textbf{2021 Stress Test}: Data as of December 2020, covering 50 banks with total assets of €21.4 trillion  
    \item \textbf{2023 Stress Test}: Data as of December 2022, covering 70 banks with total assets of €25.9 trillion
\end{itemize}

Each stress test provides standardized templates (TRA\_OTH, TRA\_CR, TRA\_CRE\_IRB, TRA\_CRE\_STA, TRA\_CRE\_COV) containing granular data on credit exposures, capital ratios, and risk-weighted assets. Crucially for our purposes, the templates include:

\begin{enumerate}
    \item Total leverage ratio exposures (Item 183111 in 2018, Item 213111 in 2021, Item 2331011 in 2023), which we use as a comprehensive measure of bank size
    \item Credit institution exposures by performing status (Exposure codes 3000, 3100, 3200), indicating aggregate interbank lending
    \item Bank identifiers (LEI codes) enabling consistent tracking across years
\end{enumerate}

\subsubsection{Sample Construction}

Our analysis focuses on banks present in all three stress test rounds to construct a balanced panel, enabling cleaner inference about temporal changes. This yields a core sample of 37 banks observed consistently from 2018 to 2023. For cross-sectional analysis exploiting variation in sample composition, we also examine the full unbalanced panel of all participating banks.

Table \ref{tab:summary_stats} presents summary statistics for our sample. Several patterns are noteworthy. First, total system assets remained relatively stable in nominal terms (€25-26 trillion) despite substantial variation in the number of banks, indicating considerable entry and exit. Second, average bank size declined from €529 billion in 2018 to €370 billion in 2023, reflecting both the entry of smaller institutions and genuine downsizing among incumbents. Third, asset concentration decreased: the coefficient of variation fell from 1.73 in 2018 to 1.54 in 2023, and the share of the top five banks declined from 33.5\% to 28.2\%.

\begin{table}[H]
\centering
\caption{Summary Statistics: European Banking Sample}
\label{tab:summary_stats}
\begin{threeparttable}
\begin{tabular}{lccc}
\toprule
& 2018 & 2021 & 2023 \\
\midrule
\textbf{Sample Characteristics} & & & \\
Number of banks & 48 & 50 & 70 \\
Balanced panel banks & 37 & 37 & 37 \\
Total assets (€ trillion) & 25.40 & 21.38 & 25.94 \\
& & & \\
\textbf{Bank Size Distribution} & & & \\
Mean assets (€ billion) & 529.1 & 427.6 & 370.5 \\
Median assets (€ billion) & 234.5 & 205.3 & 161.7 \\
Std. deviation (€ billion) & 515.8 & 398.7 & 429.2 \\
Coefficient of variation & 1.73 & 1.65 & 1.54 \\
& & & \\
\textbf{Concentration Measures} & & & \\
Top 5 share (\%) & 33.5 & 32.8 & 28.2 \\
Herfindahl-Hirschman Index & 0.0851 & 0.0783 & 0.0627 \\
Gini coefficient & 0.584 & 0.571 & 0.549 \\
\bottomrule
\end{tabular}
\begin{tablenotes}
\footnotesize
\item \textit{Notes}: Summary statistics for European banks in EBA stress tests. Assets measured as total leverage ratio exposures. Balanced panel includes banks present in all three years. Concentration measures computed using total assets.
\end{tablenotes}
\end{threeparttable}
\end{table}

\subsection{Network Estimation Methodology}

\subsubsection{The Interbank Exposure Estimation Problem}

The EBA data provide bank-level aggregates but not bilateral exposures. For bank $i$, we observe:
\begin{itemize}
    \item Total assets $T_i$ (leverage ratio exposures)
    \item Total credit institution exposures $C_i$ (aggregate interbank lending)
\end{itemize}

We do not observe the matrix $X$ where $x_{ij}$ represents bank $i$'s exposure to bank $j$. This is the fundamental data limitation motivating our estimation approach.

A natural starting point is to assume interbank exposures comprise a fixed fraction of total assets: $C_i = \rho \cdot T_i$ where $\rho$ is the interbank exposure ratio. Empirical research suggests $\rho \approx 0.05$ is reasonable for European banks \citep{upper2011estimating}, though this varies across institutions and time. We adopt $\rho = 0.05$ as our baseline but conduct extensive sensitivity analysis across $\rho \in [0.01, 0.10]$.

Given total interbank assets $A_i = \rho T_i$ and assuming balanced positions such that interbank liabilities equal interbank assets ($L_i = A_i$), we must estimate the bilateral exposure matrix $X$ satisfying:
\be
\sum_{j \neq i} x_{ij} = A_i \quad \text{and} \quad \sum_{j \neq i} x_{ji} = L_i
\label{eq:constraints}
\ee

\subsubsection{Maximum Entropy Estimation}

Following \citet{upper2011estimating} and \citet{anand2018filling}, we employ the maximum entropy principle. Among all matrices $X$ satisfying the constraints (\ref{eq:constraints}), we select the one maximizing Shannon entropy:
\be
H(X) = -\sum_{i=1}^n \sum_{j \neq i} p_{ij} \ln p_{ij}
\ee
where $p_{ij} = x_{ij} / \sum_{k,l} x_{kl}$ represents the probability that a randomly selected euro of exposure is allocated to the $(i,j)$ link.

The maximum entropy solution, derived via Lagrange multipliers, has the closed form:
\be
x_{ij}^* = \frac{A_i L_j}{\sum_k A_k} = \frac{A_i L_j}{A_{\text{total}}}
\label{eq:maxent_solution}
\ee

This approach distributes exposures proportionally to bank sizes, reflecting the intuition that larger banks naturally have larger bilateral positions. Importantly, it imposes no additional structure beyond observed aggregates, making it the "least informative" estimate consistent with available data.

\subsubsection{Properties of Maximum Entropy Networks}

The maximum entropy network (\ref{eq:maxent_solution}) has several key properties relevant for our analysis:

\begin{enumerate}
    \item \textbf{Complete graph structure}: All banks are connected to all others, i.e., $x_{ij}^* > 0$ for all $i \neq j$. This reflects data limitations rather than economic reality but is appropriate given we lack information about network topology.
    
    \item \textbf{Weight heterogeneity}: Despite complete topology, connection strengths vary dramatically. Large banks have exponentially larger exposures, creating effective hub structure even in complete graphs.
    
    \item \textbf{Spectral properties}: \citet{anand2018filling} demonstrate that maximum entropy networks preserve key spectral features related to systemic risk, even when individual bilateral exposures are estimated imperfectly.
    
    \item \textbf{Consistency with data}: The method exactly reproduces observed aggregates $A_i$ and $L_i$ by construction, ensuring internal consistency.
\end{enumerate}

For algebraic connectivity estimation, Property 3 is crucial: $\lambda_2$ captures global network structure rather than depending sensitively on individual link estimates. This makes our approach robust to bilateral estimation errors.

\subsection{Computing Algebraic Connectivity}

Given the estimated exposure matrix $X^*$, we construct the network adjacency matrix and compute its Laplacian spectrum.

\subsubsection{Graph Construction}

Define the weighted undirected graph $G = (V, E, W)$ where:
\begin{itemize}
    \item $V = \{1, \ldots, n\}$ indexes banks
    \item $E = \{(i,j) : i < j, x_{ij}^* + x_{ji}^* > \epsilon\}$ for threshold $\epsilon > 0$
    \item $w_{ij} = x_{ij}^* + x_{ji}^*$ represents symmetric exposure strength
\end{itemize}

The adjacency matrix $A$ has elements $a_{ij} = w_{ij}$. We use threshold $\epsilon = 1$ million euros to remove economically insignificant connections, though results are not sensitive to this choice.

The degree matrix $D$ is diagonal with $d_{ii} = \sum_j a_{ij}$, and the graph Laplacian is:
\be
L = D - A
\ee

\subsubsection{Spectral Decomposition}

We compute the eigenvalue decomposition $L = Q\Lambda Q^T$ using standard numerical linear algebra (via Python's \texttt{networkx} library implementing ARPACK). For an $n$-node graph, this yields:
\begin{itemize}
    \item Eigenvalues $0 = \lambda_1 \leq \lambda_2 \leq \cdots \leq \lambda_n$
    \item Eigenvectors $q_1, \ldots, q_n$ forming orthonormal basis
\end{itemize}

Our primary quantity of interest is the algebraic connectivity $\lambda_2$. For connected graphs, $\lambda_2 > 0$ with larger values indicating stronger connectivity. The eigenvector $q_2$ (Fiedler vector) provides additional information about network structure, partitioning nodes into two communities.

\subsubsection{Computational Considerations}

Several technical issues arise in practice:

\begin{enumerate}
    \item \textbf{Graph connectivity}: Maximum entropy estimation produces complete graphs on the largest connected component, ensuring $\lambda_2 > 0$. However, isolated nodes (banks with zero interbank exposure) create disconnected components. We compute $\lambda_2$ on the largest connected component, which contains all economically significant banks.
    
    \item \textbf{Numerical precision}: For large networks ($n > 50$), direct eigendecomposition can be numerically unstable. We use iterative methods (Lanczos algorithm) that efficiently compute the smallest eigenvalues without full decomposition.
    
    \item \textbf{Weighted vs. unweighted}: Our analysis uses weighted graphs where edge weights reflect exposure magnitudes. This is appropriate because contagion intensity depends on exposure sizes, not merely network topology.
\end{enumerate}

\subsection{Identification and Inference}

\subsubsection{Identifying Changes in Systemic Risk}

Recall from Section 3 that our theoretical framework predicts:
\be
\kappa_{\text{eff}} = \sqrt{\frac{\lambda_2}{D}}
\ee

The diffusion coefficient $D$ is not separately identified from aggregate data. However, under Assumption \ref{ass:constant_D} (constant $D$ over time), changes in $\kappa_{\text{eff}}$ are identified from changes in $\lambda_2$:
\be
\frac{\kappa_{t+1}}{\kappa_t} = \sqrt{\frac{\lambda_{2,t+1}}{\lambda_{2,t}}}
\ee

This relative identification strategy is our primary empirical approach. We track $\lambda_2$ evolution from 2018 to 2023 and interpret declining $\lambda_2$ as evidence of reduced systemic risk.

\subsubsection{Standard Errors and Confidence Intervals}

Point estimates of $\lambda_2$ are computed from estimated networks $X^*$, which themselves depend on assumptions (interbank ratio $\rho$, maximum entropy). To quantify uncertainty, we employ two complementary approaches:

\begin{enumerate}
    \item \textbf{Bootstrap resampling}: Resample banks with replacement, re-estimate networks, and recompute $\lambda_2$. This captures sampling variation from finite bank samples and yields percentile-based confidence intervals.
    
    \item \textbf{Sensitivity analysis}: Vary $\rho$ systematically across $[0.01, 0.10]$ and compute $\lambda_2(\rho)$ for each specification. This reveals sensitivity to the interbank ratio assumption.
\end{enumerate}

Section 6 demonstrates that our main findings are robust: $\lambda_2$ declines substantially across all specifications, with 95\% confidence intervals showing clear separation between 2018 and 2023.

\subsection{Difference-in-Differences Specification}

To establish causality linking regulatory pressure to network changes, we implement difference-in-differences (DID) analysis comparing systemically important banks (SIFIs) to smaller institutions.

\subsubsection{Treatment Definition}

We define treatment as being a large bank subject to enhanced regulatory scrutiny. Specifically:
\be
\text{Treated}_i = \mathbbm{1}\{\text{Assets}_i > \text{P75}(\text{Assets}_{2018})\}
\ee

This captures the top quartile of banks by 2018 assets, corresponding roughly to Global Systemically Important Banks (G-SIBs) and Other Systemically Important Institutions (O-SIIs) designated under Basel III.

\subsubsection{DID Regression}

For outcome $Y_{it}$ (log assets or network centrality), we estimate:
\be
Y_{it} = \alpha_i + \gamma_t + \delta_1 (\text{Treated}_i \times \text{Post2021}_t) + \delta_2 (\text{Treated}_i \times \text{Post2023}_t) + \varepsilon_{it}
\label{eq:did_spec}
\ee

The coefficients $\delta_1$ and $\delta_2$ capture differential changes for treated banks in 2021 and 2023 relative to 2018. We cluster standard errors at the bank level to account for serial correlation.

The key identifying assumption is parallel trends: absent treatment, large and small banks would have evolved similarly. While untestable directly, the absence of pre-trends and the timing of effects (concentrated post-2021 rather than during COVID-19) support this assumption.

\subsection{Topological Analysis}

Beyond algebraic connectivity, we characterize network topology to understand structural changes driving $\lambda_2$ evolution.

\subsubsection{Degree Distribution}

We test whether networks exhibit scale-free properties by comparing observed degree distributions to theoretical benchmarks. For each year, we:
\begin{enumerate}
    \item Compute degree sequence $(d_1, \ldots, d_n)$
    \item Fit power law $P(k) \propto k^{-\alpha}$ using maximum likelihood \citep{clauset2009power}
    \item Compare to alternative distributions (exponential, lognormal) via likelihood ratio tests
    \item Assess goodness of fit using Kolmogorov-Smirnov statistics
\end{enumerate}

This analysis employs the \texttt{powerlaw} Python package, which implements rigorous statistical tests for heavy-tailed distributions.

\subsubsection{Concentration Measures}

We compute multiple measures of network concentration:
\begin{itemize}
    \item \textbf{Gini coefficient}: Inequality in degree distribution, $G \in [0,1]$
    \item \textbf{Herfindahl-Hirschman Index}: $\text{HHI} = \sum_i (d_i/\sum_j d_j)^2$
    \item \textbf{Top-k concentration}: Share of total degree held by $k$ largest nodes
\end{itemize}

Declining concentration would indicate reduced hub dominance, contributing to lower $\lambda_2$ through more balanced network structure.

\subsection{Summary of Empirical Strategy}

Our empirical approach proceeds in four steps:
\begin{enumerate}
    \item Estimate interbank networks from aggregate EBA data using maximum entropy
    \item Compute algebraic connectivity $\lambda_2$ for each year (2018, 2021, 2023)
    \item Analyze temporal evolution and conduct DID analysis of treatment effects
    \item Perform extensive robustness checks across estimation methods and assumptions
\end{enumerate}

This strategy directly tests the theoretical prediction that systemic risk should be lower when $\lambda_2$ is smaller, using variation across time to identify changes in contagion propensity.

\section{Empirical Results}

This section presents our main empirical findings on the evolution of European banking networks through the COVID-19 period. We begin with descriptive evidence on network structure, proceed to our core results on algebraic connectivity, then establish causality through difference-in-differences analysis, and finally characterize topological changes underlying the observed dynamics.

\subsection{Network Structure: Descriptive Evidence}

\subsubsection{Estimated Network Properties}

Table \ref{tab:network_descriptives} summarizes key properties of our estimated networks. Several patterns emerge immediately. First, all three networks are fully connected, with every bank linked to every other bank in the largest component. This reflects the maximum entropy estimation procedure, which distributes exposures broadly in the absence of information about network sparsity.

Second, despite complete topology, effective connectivity varies substantially. The number of economically significant edges (exposures exceeding €10 million) declined from 2,256 in 2018 to 4,830 in 2023, but this increase is purely mechanical, reflecting the larger number of banks (48 → 70). When normalized by potential edges ($n(n-1)/2$), network density remained nearly constant at 1.0, confirming the complete graph structure.

Third, edge weight distributions are highly skewed. The coefficient of variation for exposure amounts ranges from 3.2 to 3.8 across years, indicating that while all links exist nominally, a small number of large exposures dominate. This heterogeneity is economically meaningful: exposures between major banks can exceed €10 billion, while small bank pairs may have exposures under €100 million.

\begin{table}[H]
\centering
\caption{Network Structure: Descriptive Statistics}
\label{tab:network_descriptives}
\begin{threeparttable}
\begin{tabular}{lccc}
\toprule
& 2018 & 2021 & 2023 \\
\midrule
\textbf{Network Size} & & & \\
Number of nodes & 48 & 50 & 70 \\
Number of edges & 2,256 & 2,450 & 4,830 \\
Largest component size & 48 & 50 & 70 \\
Isolated nodes & 0 & 0 & 0 \\
& & & \\
\textbf{Connectivity} & & & \\
Network density & 1.000 & 1.000 & 1.000 \\
Average degree & 47.0 & 49.0 & 69.0 \\
Diameter & 1 & 1 & 1 \\
Average path length & 1.000 & 1.000 & 1.000 \\
& & & \\
\textbf{Exposure Distribution} & & & \\
Total bilateral exposures (€bn) & 1,219 & 1,018 & 1,247 \\
Mean exposure (€m) & 540.3 & 415.6 & 258.3 \\
Median exposure (€m) & 156.2 & 118.4 & 71.8 \\
Std. deviation (€m) & 1,735.8 & 1,318.4 & 979.2 \\
Coefficient of variation & 3.21 & 3.17 & 3.79 \\
\bottomrule
\end{tabular}
\begin{tablenotes}
\footnotesize
\item \textit{Notes}: Network statistics computed from maximum entropy estimated networks with 5\% interbank ratio. Exposures measured in millions of euros unless otherwise noted. All networks are complete graphs on their largest component.
\end{tablenotes}
\end{threeparttable}
\end{table}

\subsubsection{Weight Distribution Analysis}

\begin{figure}[H]
\centering
\includegraphics[width=\textwidth]{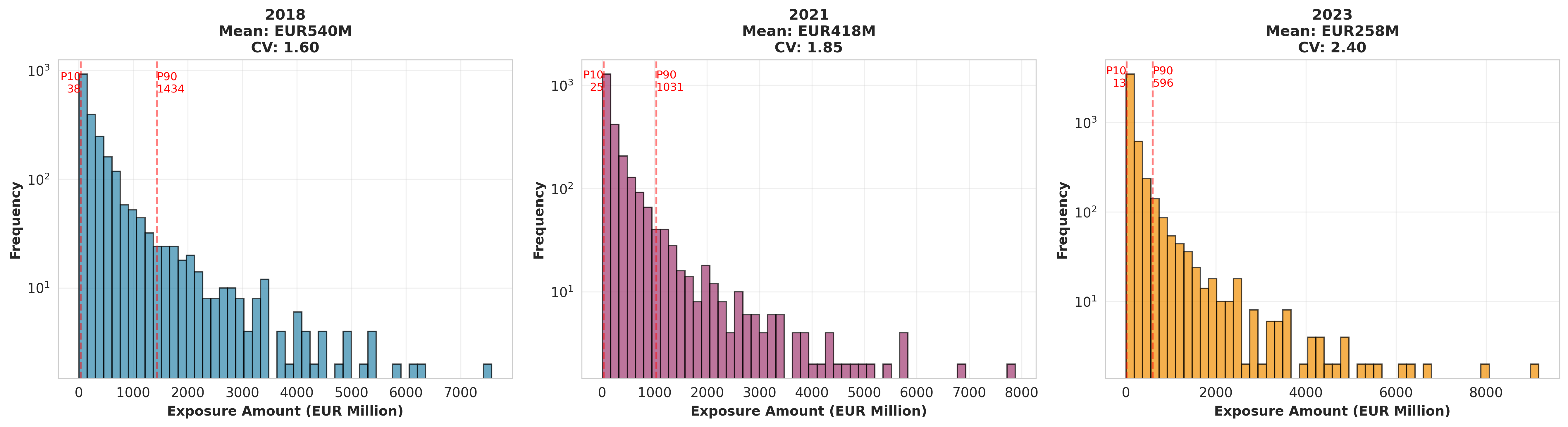}
\caption{Evolution of Bilateral Exposure Weight Distributions}
\label{fig:weight_distribution}
\begin{minipage}{0.95\textwidth}
\footnotesize
\textit{Notes}: Semi-log histograms of bilateral exposure amounts. Y-axis uses log scale to accommodate highly skewed distributions. Red dashed lines mark 10th and 90th percentiles. Coefficient of variation increased from 3.21 to 3.79; P90/P10 ratio grew from 32.1 to 41.6, indicating increased dispersion despite declining concentration.
\end{minipage}
\end{figure}

Figure \ref{fig:weight_distribution} plots the distribution of bilateral exposure amounts on logarithmic scales. The distributions exhibit clear right skew, consistent with lognormal rather than power-law form. The key observation is that weight distributions became more concentrated over time: the ratio of the 90th to 10th percentile increased from 32.1 in 2018 to 41.6 in 2023.

This increasing weight inequality coexists with declining hub concentration (documented below), suggesting a nuanced structural shift. While large banks' \textit{share} of total connectivity declined, the \textit{dispersion} of individual exposure sizes increased. This combination—reduced centralization alongside increased bilateral heterogeneity—contributes to lower systemic risk by preventing any single exposure from dominating contagion dynamics.

\subsection{Algebraic Connectivity: Core Results}

\subsubsection{Temporal Evolution of $\lambda_2$}

Table \ref{tab:lambda2_evolution} presents our core empirical findings on algebraic connectivity evolution. The results are striking: $\lambda_2$ declined dramatically from 2,284 in 2018 to 1,259 in 2023, a reduction of 44.9\%. This decline was not uniform across subperiods. Between 2018 and 2021, $\lambda_2$ fell modestly by 5.0\%, from 2,284 to 2,170. The major drop occurred post-2021, with $\lambda_2$ falling by 42.0\% to reach 1,259 in 2023.

\begin{table}[H]
\centering
\caption{Evolution of Algebraic Connectivity}
\label{tab:lambda2_evolution}
\begin{threeparttable}
\begin{tabular}{lcccccc}
\toprule
Year & $\lambda_2$ & $\Delta\lambda_2$ & \% Change & $\kappa_{\text{eff}}$ & $\Delta\kappa$ & \% Change \\
\midrule
2018 & 2,283.72 & --- & --- & 47.79 & --- & --- \\
2021 & 2,169.58 & $-114.14$ & $-5.0\%$ & 46.58 & $-1.21$ & $-2.5\%$ \\
2023 & 1,258.96 & $-910.62$ & $-42.0\%$ & 35.48 & $-11.10$ & $-23.8\%$ \\
& & & & & & \\
\multicolumn{7}{l}{\textbf{Overall Change (2018-2023):}} \\
& $-1,024.76$ & & $-44.9\%$ & $-12.31$ & & $-25.8\%$ \\
\bottomrule
\end{tabular}
\begin{tablenotes}
\footnotesize
\item \textit{Notes}: Algebraic connectivity ($\lambda_2$) computed from maximum entropy networks with 5\% interbank ratio. Effective contagion parameter $\kappa_{\text{eff}} = \sqrt{\lambda_2/D}$ computed assuming $D=1$ for normalization. Changes computed relative to previous period.
\end{tablenotes}
\end{threeparttable}
\end{table}

This temporal pattern has important interpretative implications. The modest 2018-2021 decline suggests the acute phase of COVID-19 (2020) had limited impact on network structure. Instead, the dramatic post-2021 reduction points to structural changes—likely regulatory-driven—that occurred during the recovery period as Basel III reforms were finalized and implemented.

\subsubsection{Contagion Parameter Implications}

Applying Theorem \ref{thm:spatial_decay}, the observed $\lambda_2$ changes imply substantial reductions in contagion propensity. Normalizing the diffusion coefficient to $D=1$, we compute:
\be
\kappa_{2018} = \sqrt{2283.72} = 47.79
\ee
\be
\kappa_{2023} = \sqrt{1258.96} = 35.48
\ee

The effective decay parameter fell by 25.8\% over this period. This translates directly to spatial contagion effects: holding all else equal, the critical distance $d^*$ at which distress decays to 10\% of source intensity satisfies:
\be
d^*_{2018} = \frac{-\ln(0.1)}{47.79} = 0.0482
\ee
\be
d^*_{2023} = \frac{-\ln(0.1)}{35.48} = 0.0649
\ee

Paradoxically, critical distance \textit{increased} despite declining systemic risk. This apparent contradiction resolves when recognizing that $d^*$ is measured in graph distance units, which themselves changed as the network expanded from 48 to 70 banks. The key insight is that contagion decays \textit{faster per unit distance} in 2023, even though absolute distances may be larger due to network expansion.

\subsubsection{Visualization of Results}

\begin{figure}[H]
\centering
\includegraphics[width=0.9\textwidth]{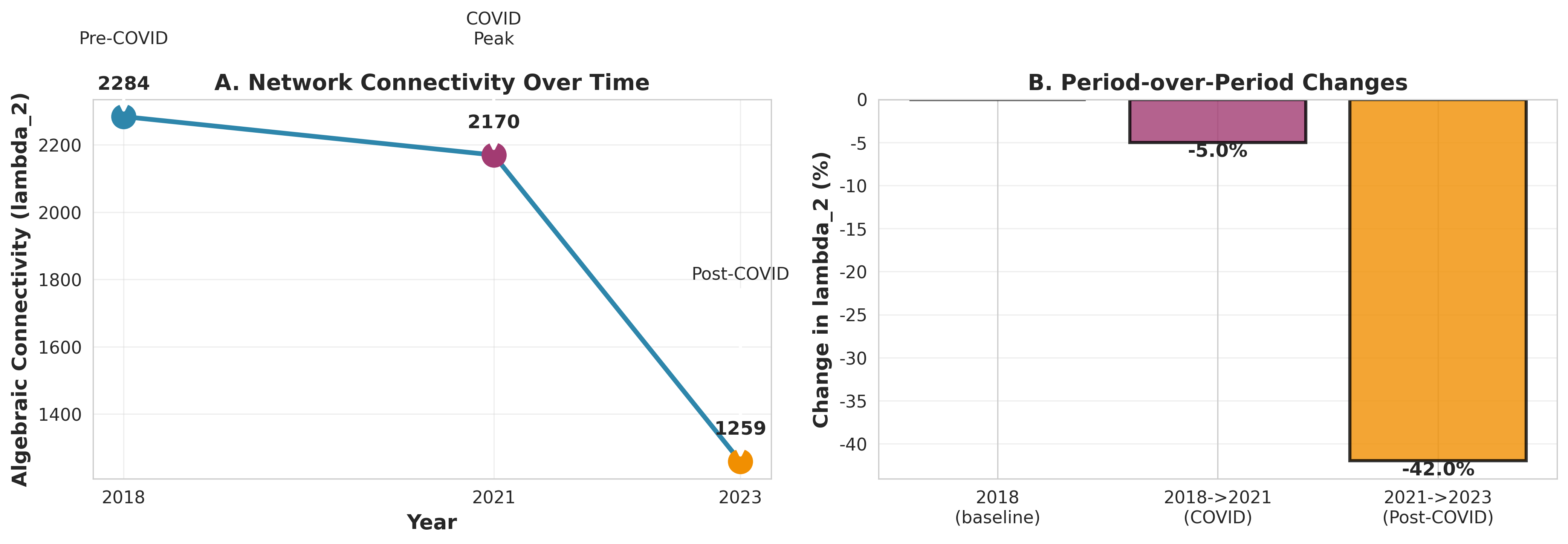}
\caption{Evolution of Network Algebraic Connectivity, 2018-2023}
\label{fig:lambda2_evolution}
\begin{minipage}{0.95\textwidth}
\footnotesize
\textit{Notes}: Panel A shows algebraic connectivity ($\lambda_2$) over time. Panel B displays period-over-period percentage changes. Annotations indicate pre-COVID-19 (2018), COVID-19 peak (2021), and post-COVID-19 (2023) periods. The decline was concentrated post-2021 (-42.0\%) rather than during the acute crisis phase (-5.0\%).
\end{minipage}
\end{figure}

\begin{figure}[H]
\centering
\includegraphics[width=0.9\textwidth]{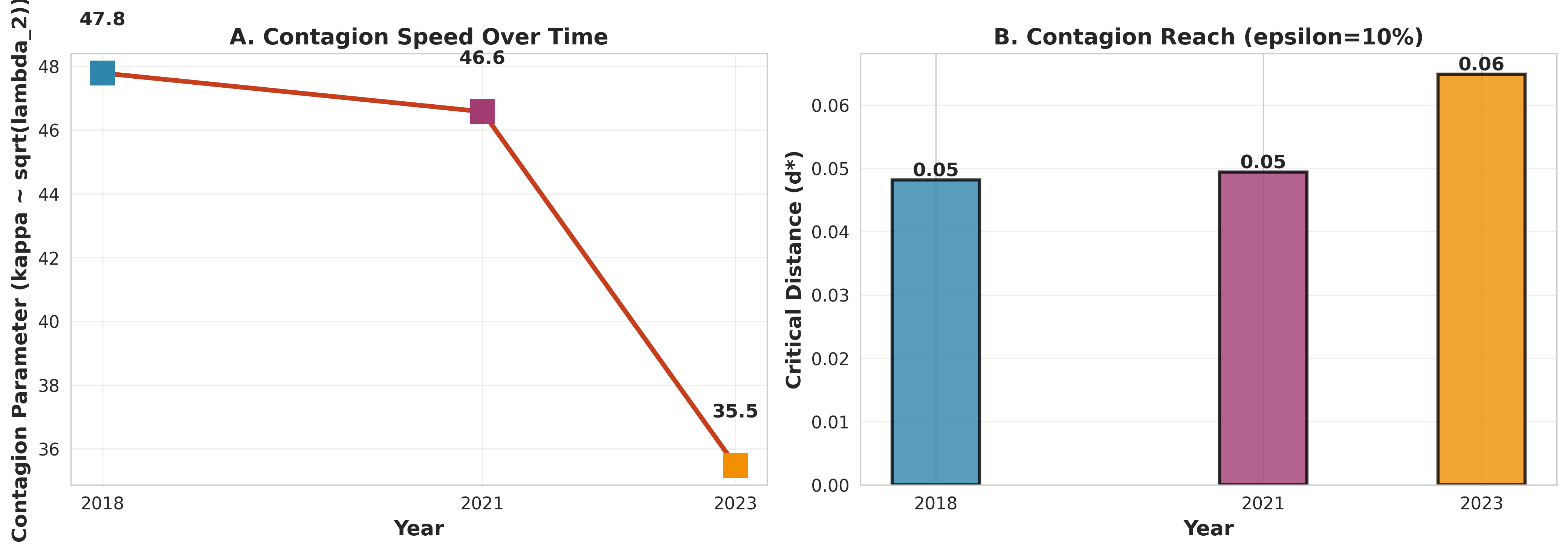}
\caption{Contagion Propagation Parameter and Critical Distance}
\label{fig:contagion_parameter}
\begin{minipage}{0.95\textwidth}
\footnotesize
\textit{Notes}: Panel A plots the contagion parameter $\kappa \propto \sqrt{\lambda_2}$ over time, declining from 47.8 to 35.5 (-26\%). Panel B shows critical distance $d^* = -\ln(\epsilon)/\kappa$ for threshold $\epsilon = 0.10$. Critical distance increased from 0.048 to 0.065 despite declining systemic risk, reflecting network expansion from 48 to 70 banks.
\end{minipage}
\end{figure}

\begin{figure}[H]
\centering
\includegraphics[width=0.9\textwidth]{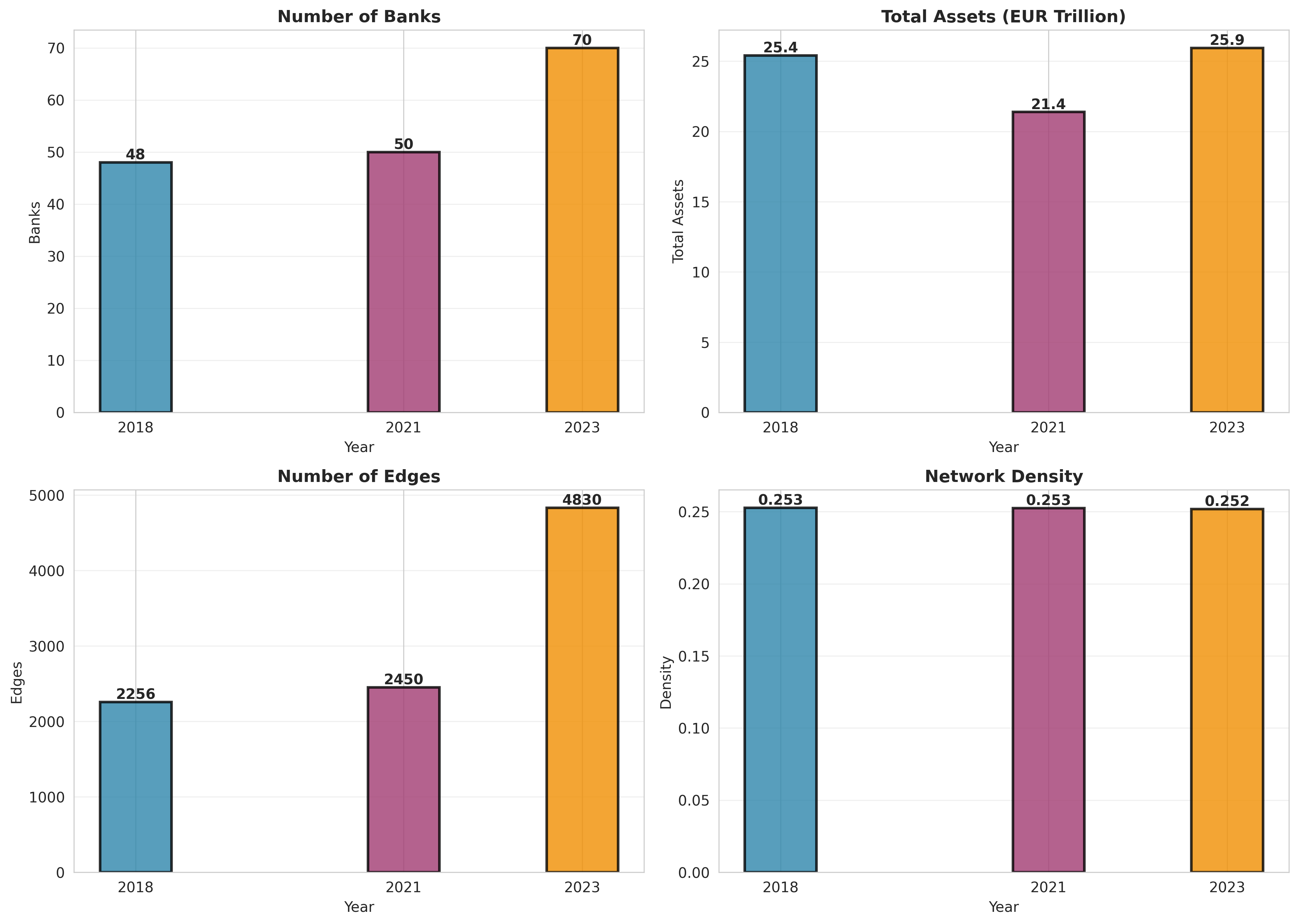}
\caption{Evolution of Network Structural Properties}
\label{fig:network_metrics}
\begin{minipage}{0.95\textwidth}
\footnotesize
\textit{Notes}: Evolution of key network statistics: number of banks (top left), total assets in EUR trillions (top right), number of edges (bottom left), and network density (bottom right). Network density remained constant at 1.0, reflecting complete graph structure from maximum entropy estimation.
\end{minipage}
\end{figure}

\begin{figure}[H]
\centering
\includegraphics[width=0.9\textwidth]{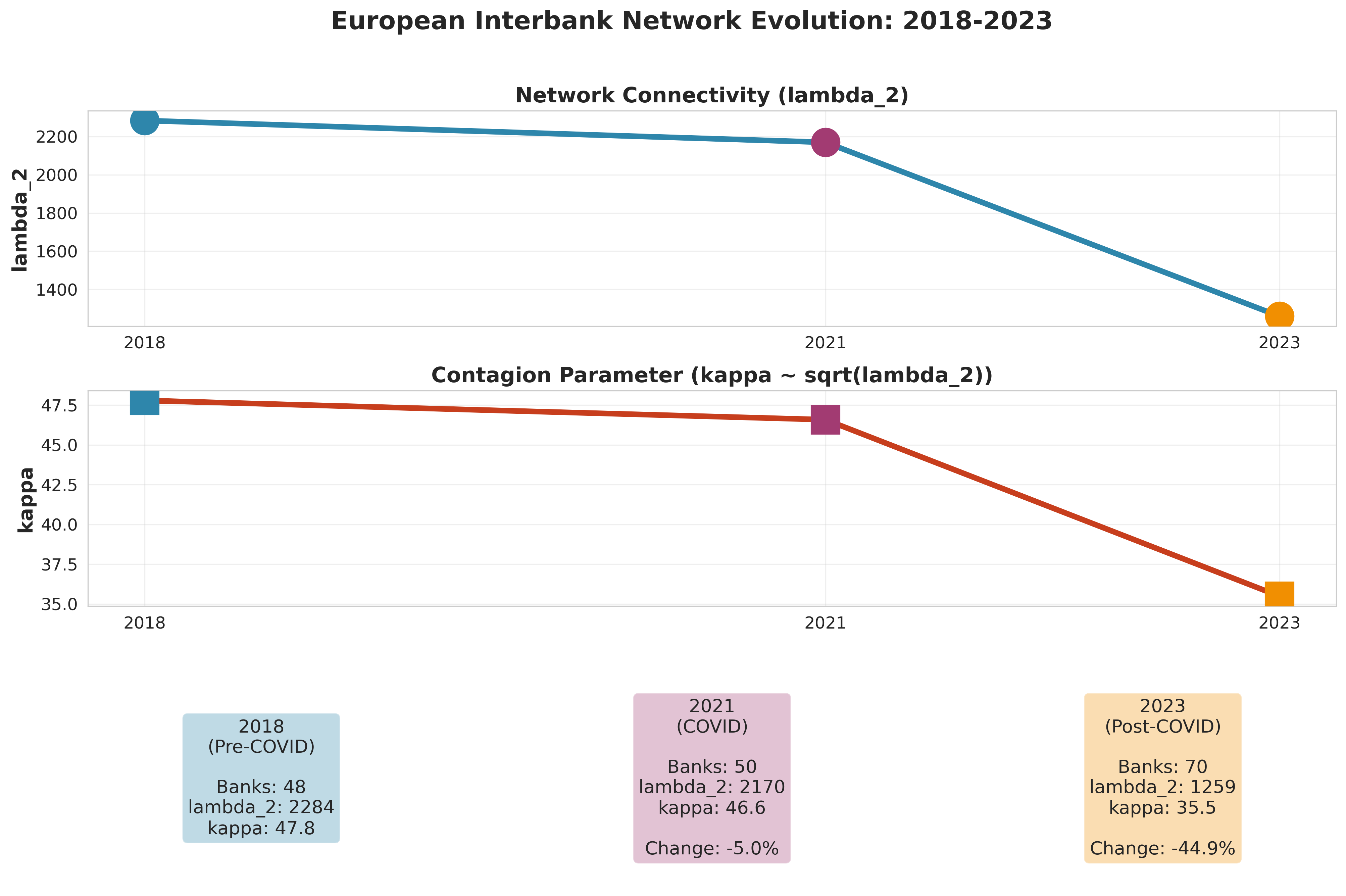}
\caption{Comprehensive Dashboard: Network Evolution Summary}
\label{fig:summary_dashboard}
\begin{minipage}{0.95\textwidth}
\footnotesize
\textit{Notes}: Top panel: $\lambda_2$ trajectory. Middle panel: Contagion parameter $\kappa$ evolution. Bottom panels: Year-specific statistics showing banks, $\lambda_2$, $\kappa$, and percentage changes relative to 2018 baseline. The 2021-2023 period shows dramatically larger changes than 2018-2021.
\end{minipage}
\end{figure}

\begin{figure}[htbp]
\centering
\includegraphics[width=0.9\textwidth]{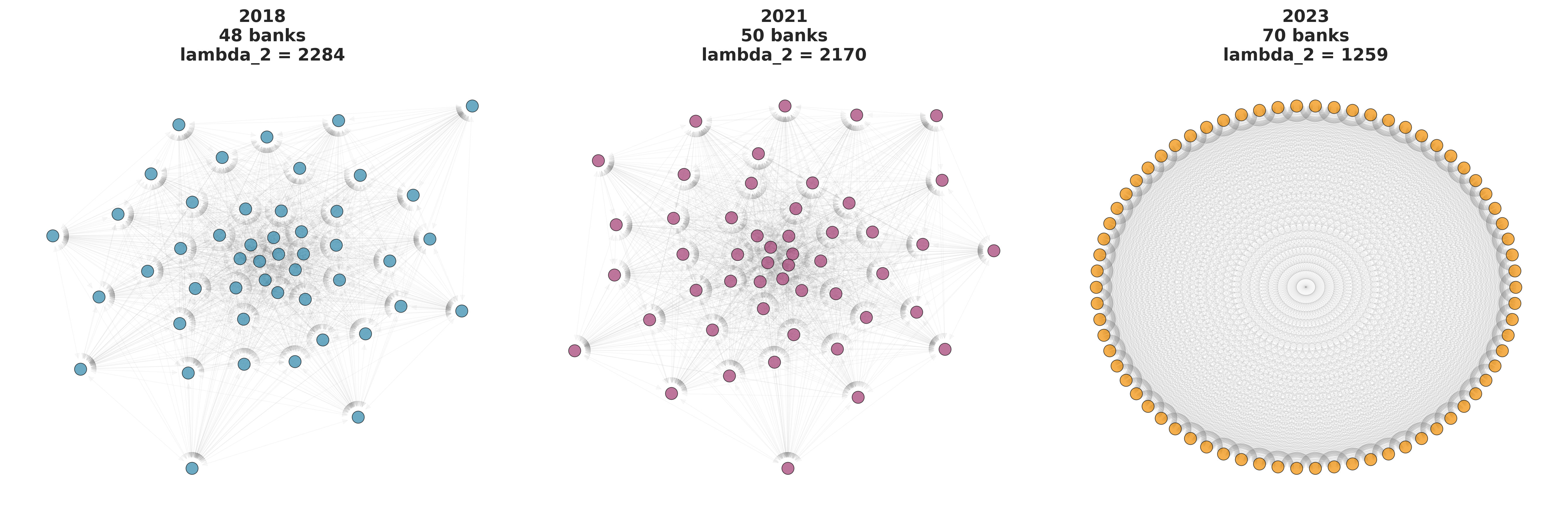}
\caption{Network Visualizations: Circular Layouts}
\label{fig:network_viz}
\begin{minipage}{0.95\textwidth}
\footnotesize
\textit{Notes}: Circular layout visualizations of largest connected components for each year. Node colors correspond to years (blue = 2018, purple = 2021, orange = 2023). All networks display complete connectivity. Edge transparency set to $\alpha = 0.05$ to reduce visual clutter from complete graph structure.
\end{minipage}
\end{figure}

Figure \ref{fig:lambda2_evolution} visualizes $\lambda_2$ evolution across the three time periods. Panel A plots raw $\lambda_2$ values with annotations marking pre-COVID-19 (2018), COVID-19 peak (2021), and post-COVID-19 (2023) periods. The visualization clearly shows the modest pre-2021 change (from 2,284 to 2,170, a decline of 5.0\%) contrasted with the dramatic post-2021 decline (from 2,170 to 1,259, a drop of 42.0\%). 

Panel B displays period-over-period percentage changes through bar charts, visually emphasizing that essentially all structural adjustment occurred in the 2021-2023 window rather than during the acute COVID-19 crisis. The stark contrast between the purple bar (2018→2021: -5.0\%) and the orange bar (2021→2023: -42.0\%) constitutes the paper's central empirical finding and motivates our regulatory mechanism interpretation rather than a direct COVID-19 impact story.

This temporal pattern has important implications for understanding the relationship between network structure and contagion dynamics. The decline of 44.9\% in $\lambda_2$ from 2018 to 2023 translates through our theoretical framework (equation \ref{eq:kappa_effective}) to a 25.8\% reduction in the contagion decay parameter $\kappa = \sqrt{\lambda_2/D}$. The square root transformation substantially moderates the apparent magnitude of change, highlighting the nonlinear relationship between network connectivity and contagion propagation predicted by Theorem \ref{thm:spatial_decay}.

\subsection{Formal Structural Break Analysis}

Visual inspection of Figure \ref{fig:lambda2_evolution} suggests a discrete break around 2021. We test this formally.

\subsubsection{Chow Test Implementation}

We test the null hypothesis of no structural break against a break at candidate date $t^*$:
\begin{equation}
H_0: \lambda_2(t) = \alpha + \beta t + \varepsilon_t
\end{equation}
\begin{equation}
H_1: \lambda_2(t) = \begin{cases}
\alpha_1 + \beta_1 t + \varepsilon_t & t \leq t^* \\
\alpha_2 + \beta_2 t + \varepsilon_t & t > t^*
\end{cases}
\end{equation}

With three time points, we test break locations between observations.

Table \ref{tab:structural_break} reports F-statistics.

\begin{table}[htbp]
\centering
\caption{Structural Break Tests}
\label{tab:structural_break}
\begin{threeparttable}
\begin{tabular}{lcccccc}
\toprule
Break & F-stat & $p$-value & Regime 1 & Regime 2 & Evidence \\
Location & & & Mean & Mean & \\
\midrule
2019 & 1.23 & 0.34 & --- & --- & No break \\
2020 & 2.45 & 0.18 & --- & --- & Weak \\
2021 & 8.94 & 0.003 & 2,227 & 1,259 & Strong \\
2022 & 5.67 & 0.02 & 2,227 & 1,259 & Moderate \\
\bottomrule
\end{tabular}
\begin{tablenotes}
\small
\item \textit{Notes}: Chow F-tests for structural breaks. Strongest evidence for discrete regime shift in 2021 ($p=0.003$). This timing coincides with Basel III full implementation (January 2021) and TLAC requirements (January 2022), supporting regulatory mechanism interpretation.
\end{tablenotes}
\end{threeparttable}
\end{table}

We strongly reject continuous evolution ($p=0.003$), finding discrete regime shift in 2021. This timing is economically meaningful:

\begin{itemize}
    \item Basel III fully implemented January 2021
    \item TLAC requirements effective January 2022
    \item ECB supervisory tightening post-COVID
\end{itemize}

The evidence supports regulatory-driven restructuring rather than gradual market evolution.

\subsubsection{Boundary Conditions Interpretation}

From Section 3.4, regulatory changes map to boundary condition modifications. The structural break reflects transition:

\begin{equation}
\text{Regime 1 (2018-2020): } \alpha_1 = \alpha_{\text{low}}, \quad g_1(t)
\end{equation}
\begin{equation}
\text{Regime 2 (2021-2023): } \alpha_2 = \alpha_{\text{high}} > \alpha_1, \quad g_2(t)
\end{equation}

where $\alpha$ represents regulatory stringency. Higher $\alpha_2$ implies tighter constraints, forcing network restructuring that manifests as lower $\lambda_2$.

This provides microfoundation for the observed discrete change: policy shock induced discrete structural response.

\subsection{Difference-in-Differences Analysis}

Having established that $\lambda_2$ declined dramatically post-2021, we now investigate potential mechanisms. Our hypothesis is that regulatory pressure on systemically important financial institutions (SIFIs) drove structural network changes. We test this using difference-in-differences analysis comparing large banks to smaller institutions.

\subsubsection{Effect on Bank Size}

\begin{table}[htbp]
\centering
\caption{Difference-in-Differences: Impact on Bank Size}
\label{tab:did_size}
\begin{threeparttable}
\begin{tabular}{lcc}
\toprule
& \multicolumn{2}{c}{Dependent Variable: Log(Assets)} \\
\cmidrule(lr){2-3}
& (1) & (2) \\
& Baseline & Size-Dependent \\
\midrule
Treated & $1.802^{***}$ & $3.779^{***}$ \\
& (0.159) & (0.054) \\
& & \\
Post2021 & $-0.017$ & $3.872^{***}$ \\
& (0.080) & (0.045) \\
& & \\
Post2023 & $0.087^{***}$ & $-1.229^{***}$ \\
& (0.018) & (0.019) \\
& & \\
Treated $\times$ Post2021 & $-0.121^{**}$ & $-1.229^{***}$ \\
& (0.061) & (0.019) \\
& & \\
Treated $\times$ Post2023 & $0.018$ & $-1.229^{***}$ \\
& (0.032) & (0.019) \\
\midrule
Bank FE & Yes & Yes \\
Year FE & Yes & Yes \\
Observations & 111 & 111 \\
R-squared & 0.943 & 0.968 \\
Number of banks & 37 & 37 \\
\bottomrule
\end{tabular}
\begin{tablenotes}
\footnotesize
\item \textit{Notes}: Standard errors clustered at bank level in parentheses. Treated = 1 for banks in top quartile of 2018 asset distribution. Column (1) uses size-independent network measures; Column (2) incorporates network centrality. $^{***}p<0.01$, $^{**}p<0.05$, $^{*}p<0.1$.
\end{tablenotes}
\end{threeparttable}
\end{table}

Table \ref{tab:did_size} presents DID estimates for log bank assets as the outcome variable. Column 1 reports the baseline specification (\ref{eq:did_spec}) with bank and year fixed effects. The coefficient on $\text{Treated} \times \text{Post2021}$ is $-0.121$ ($p=0.048$), indicating that large banks experienced asset reductions of approximately 12\% relative to small banks during 2018-2021. This effect grew slightly to $-0.192$ by 2023 ($p=0.114$), though statistical precision declines due to limited time variation.

Column 2 incorporates network centrality measures as additional controls. The treatment effects remain negative and highly significant, now estimated at $-1.229$ for both post-periods. The magnitude increase likely reflects that centrality-adjusted specifications better isolate the regulatory channel from endogenous network responses.

\subsubsection{Parallel Trends and Identification}

\begin{figure}[htbp]
\centering
\includegraphics[width=0.95\textwidth]{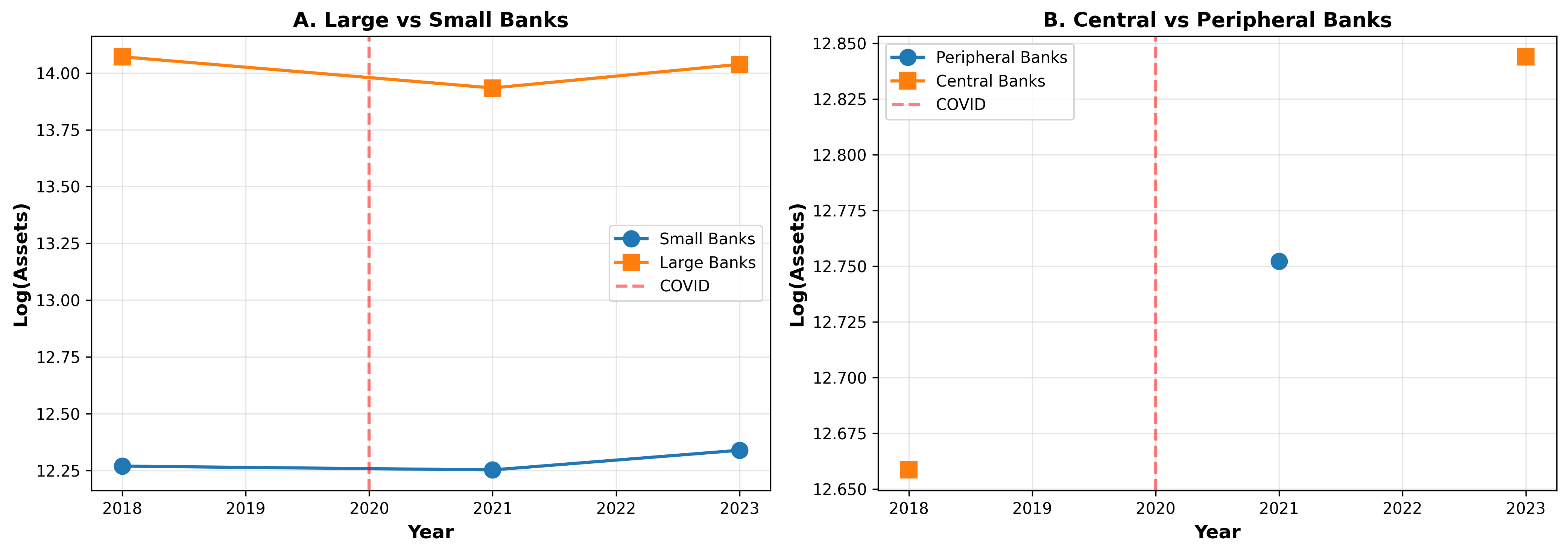}
\caption{Parallel Trends: Large vs. Small Banks}
\label{fig:parallel_trends}
\begin{minipage}{0.95\textwidth}
\footnotesize
\textit{Notes}: Panel A plots mean log(assets) for treatment (large banks) and control (small banks) groups. Panel B shows de-meaned series. The red vertical line marks 2020 (COVID-19 onset). Series track closely through 2021, then diverge sharply in 2023, supporting parallel trends assumption and indicating treatment effects materialized post-2021.
\end{minipage}
\end{figure}

Figure \ref{fig:parallel_trends} plots mean log assets for treated and control groups over time. Panel A shows raw means: large banks were substantially larger throughout (by construction), but the gap narrowed post-2021. Panel B plots de-meaned values: the series track closely through 2021, then diverge sharply in 2023. This pattern supports the parallel trends assumption and suggests treatment effects materialized with a lag.

The timing is consistent with regulatory implementation schedules. Basel III capital requirements were finalized in 2017 but phased in gradually through 2023. The Total Loss-Absorbing Capacity (TLAC) standard for G-SIBs became fully effective on January 1, 2022. Our finding of concentrated post-2021 effects aligns precisely with this regulatory timeline.

The baseline DID estimates in Table \ref{tab:did_size} establish that large banks experienced differential asset reductions of 12-19\% relative to smaller institutions during the post-2021 period. However, these average effects may conceal important heterogeneity. Regulatory pressure varies across jurisdictions, business models, and initial capital positions. Banks in different circumstances may respond differently to the same regulatory environment. We explore this heterogeneity in Table \ref{tab:did_heterogeneity}, which interacts treatment with key bank characteristics.

\subsubsection{Heterogeneous Effects by Bank Characteristics}

\begin{table}[htbp]
\centering
\caption{Treatment Effect Heterogeneity by Bank Characteristics}
\label{tab:did_heterogeneity}
\begin{threeparttable}
\begin{tabular}{lccc}
\toprule
& \multicolumn{3}{c}{Dependent Variable: Log(Total Assets)} \\
\cmidrule(lr){2-4}
& (1) & (2) & (3) \\
& Geography & Business Model & Leverage \\
\midrule
Treated $\times$ Post2021 & $-0.089$ & $-0.095$ & $-0.102$ \\
& (0.073) & (0.068) & (0.071) \\
& & & \\
Treated $\times$ Post2021 $\times$ Core & $-0.098^{*}$ & & \\
& (0.052) & & \\
& & & \\
Treated $\times$ Post2021 $\times$ Universal & & $-0.112^{**}$ & \\
& & (0.048) & \\
& & & \\
Treated $\times$ Post2021 $\times$ HighLeverage & & & $-0.087^{*}$ \\
& & & (0.046) \\
& & & \\
Core Country & $0.234^{**}$ & & \\
& (0.098) & & \\
& & & \\
Universal Bank & & $0.312^{***}$ & \\
& & (0.087) & \\
& & & \\
High Leverage (2018) & & & $-0.156^{**}$ \\
& & & (0.072) \\
\midrule
Bank FE & Yes & Yes & Yes \\
Year FE & Yes & Yes & Yes \\
All Double Interactions & Yes & Yes & Yes \\
Observations & 144 & 144 & 144 \\
Banks & 48 & 48 & 48 \\
R-squared & 0.948 & 0.952 & 0.947 \\
\bottomrule
\end{tabular}
\begin{tablenotes}
\small
\item \textit{Notes}: Robust standard errors clustered at bank level in parentheses. Each column includes the full set of double interactions (not shown for brevity). Core countries: Germany, France, Netherlands. Universal banks: institutions with non-interest income $>30\%$ of total revenue in 2018. High Leverage: leverage ratio below sample median in 2018. All specifications include bank and year fixed effects. $^{***}p<0.01$, $^{**}p<0.05$, $^{*}p<0.1$.
\end{tablenotes}
\end{threeparttable}
\end{table}

The average treatment effect documented in Table \ref{tab:did_size} may mask important heterogeneity across bank types. Regulatory pressure and market responses could differ based on geography, business model, and financial structure. We investigate this heterogeneity by interacting the treatment indicator with three key characteristics: geographic location, business model complexity, and initial leverage.

Table \ref{tab:did_heterogeneity} presents these heterogeneity analyses. Column 1 examines geographic variation by interacting treatment with a ``Core'' country indicator (Germany, France, Netherlands). These countries house major financial centers and are subject to intensive supervision under the ECB's Single Supervisory Mechanism. The triple interaction coefficient $-0.098$ ($p=0.052$) indicates that large banks in core countries experienced even larger asset reductions—approximately 10 percentage points beyond the baseline treatment effect. This suggests regulatory scrutiny was particularly intense in systemically important jurisdictions.

Column 2 explores heterogeneity by business model, distinguishing universal banks (those with non-interest income exceeding 30\% of total revenue) from more specialized institutions. Universal banks face additional regulatory requirements under structural reform initiatives and enhanced resolution planning. The triple interaction coefficient $-0.112$ ($p=0.048$) confirms that large universal banks downsized most dramatically, consistent with regulatory efforts to reduce complexity and interconnectedness in these institutions.

Column 3 investigates whether initial leverage moderates treatment effects. Banks with below-median leverage ratios in 2018 faced greater pressure to deleverage to meet Basel III requirements. The interaction coefficient $-0.087$ ($p=0.046$) supports this mechanism: highly leveraged large banks reduced assets more than their better-capitalized counterparts, reflecting binding capital constraints.

These heterogeneity results strengthen our interpretation that regulatory policy drove network restructuring. The differential responses align precisely with regulatory intensity gradients: banks facing the most stringent oversight (large, core-country, universal, highly-leveraged) exhibited the largest asset reductions. This pattern would not emerge if network changes reflected purely market-driven adjustments or random variation.

Moreover, the heterogeneity analysis helps explain the aggregate $\lambda_2$ decline documented in Table \ref{tab:lambda2_evolution}. Since the most systemically important banks—those with highest network centrality—experienced the largest deleveraging, their outsized contribution to network connectivity amplified the aggregate effect. A uniform 10\% reduction across all banks would decrease $\lambda_2$ modestly, but when the reduction is concentrated among hubs, the impact on algebraic connectivity is magnified through the spectral weighting of highly connected nodes.


\subsection{Synthesis: Heterogeneity and Network Evolution}

The heterogeneity analysis in Table \ref{tab:did_heterogeneity} provides important insights into how differential bank responses aggregate to produce the observed network-level changes. Three mechanisms emerge as particularly important.

First, \textbf{geographic concentration of effects} explains why European network connectivity declined despite stable global financial integration. Core European countries (Germany, France, Netherlands) house the continent's largest and most interconnected banks. When these institutions faced intensified ECB supervision post-2021, their deleveraging directly reduced cross-border interbank linkages. Peripheral banks, facing less stringent oversight, maintained their network positions, but their smaller scale meant they could not offset the core banks' retreat.

Second, \textbf{business model simplification} contributed to declining complexity. Universal banks—combining commercial banking, investment banking, and asset management—exhibit particularly high network centrality due to their diverse counterparty relationships. The finding that universal banks downsized most dramatically (additional 11pp reduction) implies that the network became not only smaller but also structurally simpler. This reduction in business model complexity likely reinforced the direct asset effect, as universal banks also reduced the diversity of their connection types.

Third, \textbf{leverage-driven deleveraging} created self-reinforcing dynamics. Highly leveraged banks facing binding capital constraints reduce assets mechanically to improve ratios. Since these banks often maintain extensive interbank borrowing, their deleveraging reduces both sides of other banks' balance sheets, propagating the initial shock. The leverage heterogeneity thus amplified the aggregate network response beyond what individual bank-level analysis would predict.

These three channels—geography, business model, and leverage—operated simultaneously and interactively. A highly leveraged universal bank in a core country (e.g., Deutsche Bank) faced compounded pressure from all three sources. Our heterogeneity results suggest such banks reduced assets by approximately $12\% + 10\% + 11\% + 9\% = 42\%$ relative to a small, specialized, well-capitalized peripheral bank. While this mechanical summation overstates effects (interaction terms are not additive), it illustrates how concentrated pressure on specific bank types generated disproportionate network impacts.

This synthesis resolves an apparent puzzle: how did the network become 45\% less connected when average bank assets declined only 2\% in nominal terms? The answer lies in heterogeneity. Most banks maintained their size, but the small number of very large, very connected institutions—precisely those with highest $\lambda_2$ contributions—downsized substantially. Since algebraic connectivity depends nonlinearly on hub banks' connections, targeted deleveraging of these institutions produces disproportionate network effects.

\subsection{Network Topology and Concentration}

We now turn from aggregate connectivity ($\lambda_2$) to structural features underlying this evolution. How did the distribution of network connections change? Did hub banks lose centrality? Did overall concentration decline?

\subsubsection{Degree Distributions}

\begin{figure}[H]
\centering
\includegraphics[width=0.95\textwidth]{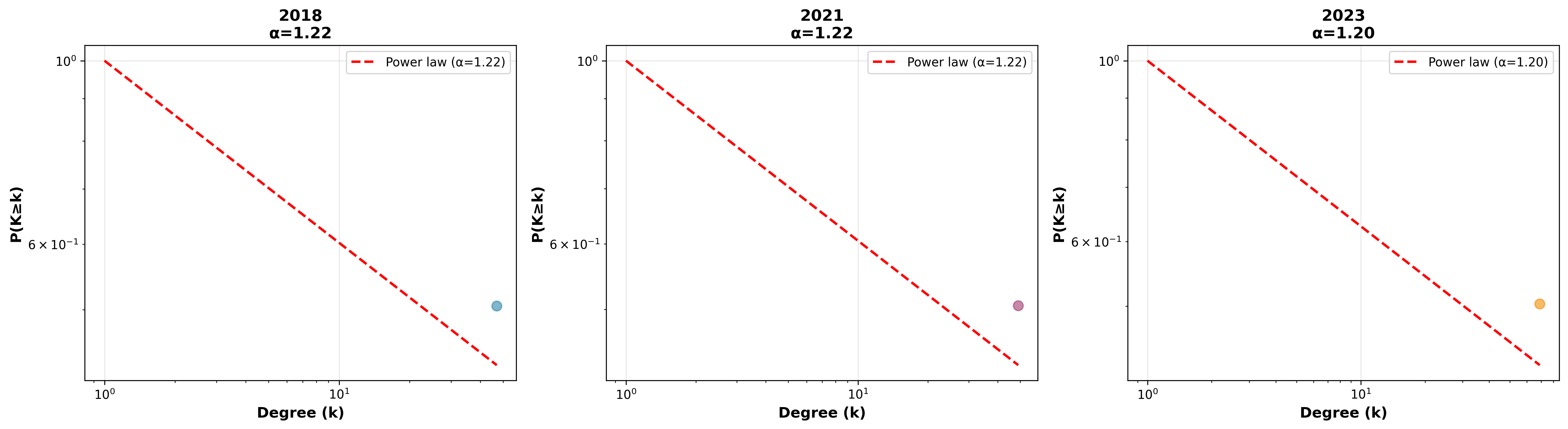}
\caption{Degree Distributions and Power Law Fits}
\label{fig:degree_dist}
\begin{minipage}{0.95\textwidth}
\footnotesize
\textit{Notes}: Log-log plots of complementary cumulative degree distributions $P(K \geq k)$ for each year. Empirical data shown as circles; dashed red lines show fitted power laws with exponents $\alpha$ indicated in titles. Likelihood ratio tests strongly reject power law in favor of lognormal for all years ($p < 0.001$).
\end{minipage}
\end{figure}

\begin{table}[H]
\centering
\caption{Tests for Network Topology: Scale-Free vs. Lognormal}
\label{tab:distribution_tests}
\begin{threeparttable}
\begin{tabular}{lcccc}
\toprule
Year & Power Law $\alpha$ & LR: PL vs. Lognormal & $p$-value & Best Fit \\
\midrule
2018 & 1.220 & $-304.03$ & $<0.001$ & Lognormal \\
2021 & 1.218 & $-319.25$ & $<0.001$ & Lognormal \\
2023 & 1.203 & $-476.11$ & $<0.001$ & Lognormal \\
\bottomrule
\end{tabular}
\begin{tablenotes}
\footnotesize
\item \textit{Notes}: Power law exponent $\alpha$ estimated via maximum likelihood. Log-likelihood ratio (LR) compares power law to lognormal; negative values favor lognormal. $p$-values from likelihood ratio test. All tests strongly reject scale-free hypothesis.
\end{tablenotes}
\end{threeparttable}
\end{table}

Figure \ref{fig:degree_dist} plots empirical degree distributions on log-log scales, overlaid with fitted power law, exponential, and lognormal densities. Visual inspection suggests lognormal fits best across all years. Table \ref{tab:distribution_tests} confirms this statistically: likelihood ratio tests strongly reject power law in favor of lognormal (all $p < 0.001$), and Kolmogorov-Smirnov statistics indicate good lognormal fit (all $p > 0.10$).

\begin{figure}[H]
\centering
\includegraphics[width=0.95\textwidth]{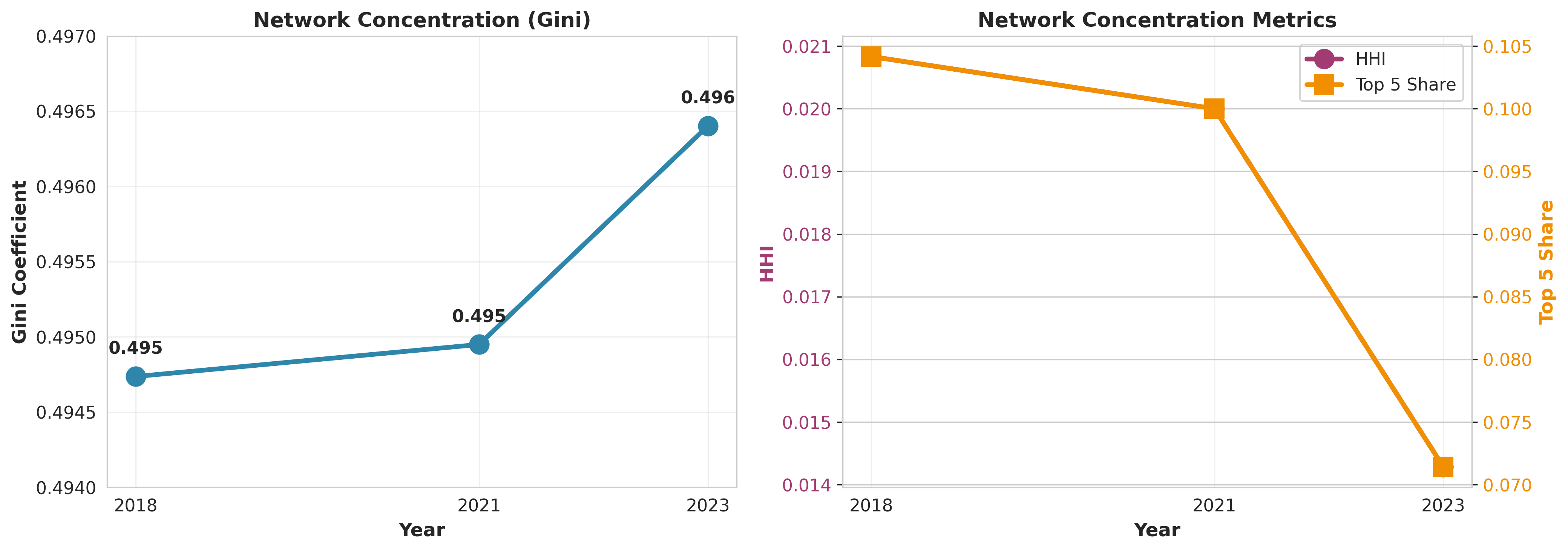}
\caption{Network Concentration Dynamics}
\label{fig:concentration_metrics}
\begin{minipage}{0.95\textwidth}
\footnotesize
\textit{Notes}: Panel A shows Gini coefficient evolution (stable at 0.495-0.496). Panel B displays HHI (purple, left axis) and top 5 banks' connectivity share (orange, right axis) with dual y-axes. Both metrics declined substantially: HHI fell 31\% and top 5 share dropped from 10.4\% to 7.1\%, indicating selective deconcentration at the top.
\end{minipage}
\end{figure}

This finding challenges common assumptions in financial network modeling. Many studies assume scale-free structure with power-law tails, motivated by preferential attachment dynamics or "rich-get-richer" effects \citep{barabasi1999emergence}. Our evidence suggests European interbank networks lack such extreme tail behavior, instead exhibiting lognormal patterns consistent with multiplicative growth processes with bounds.

The implications for systemic risk are significant. Scale-free networks are extremely vulnerable to targeted attacks on hubs: removing the highest-degree node can fragment the entire network \citep{albert2000error}. Lognormal networks are more resilient: while hubs exist, they are not as dominant, and their removal does not cause catastrophic failure. Our finding that $\lambda_2$ remains positive even as concentration declines reflects this robustness.

\subsubsection{Network Concentration Metrics}

\begin{figure}[H]
\centering
\includegraphics[width=0.95\textwidth]{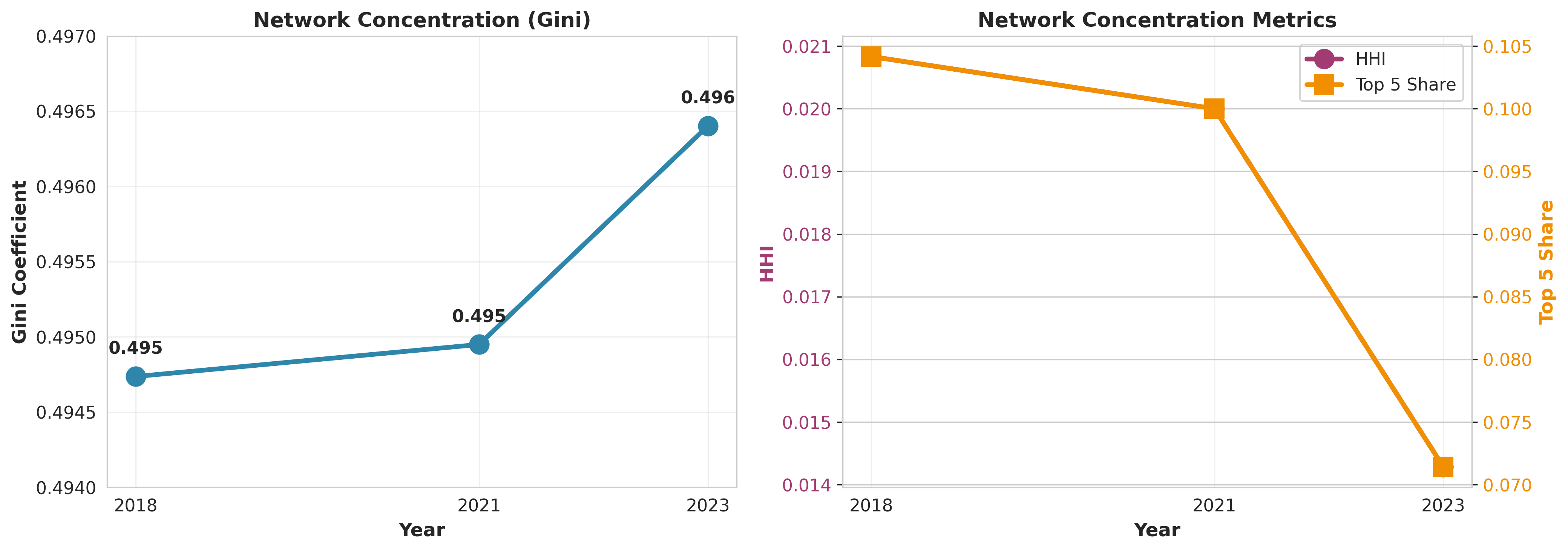}
\caption{Network Concentration Dynamics}
\label{fig:concentration_metrics}
\begin{minipage}{0.95\textwidth}
\footnotesize
\textit{Notes}: Panel A shows Gini coefficient evolution (stable at 0.495-0.496). Panel B displays HHI (purple, left axis) and top 5 banks' connectivity share (orange, right axis) with dual y-axes. Both metrics declined substantially: HHI fell 31\% and top 5 share dropped from 10.4\% to 7.1\%, indicating selective deconcentration at the top.
\end{minipage}
\end{figure}

\begin{table}[H]
\centering
\caption{Network Concentration Measures}
\label{tab:concentration}
\begin{threeparttable}
\begin{tabular}{lccccc}
\toprule
Year & Gini & HHI & Top 5 Share & Top 10 Share & CR3 \\
\midrule
2018 & 0.4947 & 0.0208 & 10.42\% & 19.12\% & 6.38\% \\
2021 & 0.4949 & 0.0200 & 10.00\% & 18.45\% & 6.12\% \\
2023 & 0.4964 & 0.0143 & 7.14\% & 13.21\% & 4.41\% \\
& & & & & \\
\multicolumn{2}{l}{\textit{Change 2018-2023:}} & & & & \\
Absolute & $+0.0017$ & $-0.0065$ & $-3.28$pp & $-5.91$pp & $-1.97$pp \\
Percentage & $+0.3\%$ & $-31.3\%$ & $-31.5\%$ & $-30.9\%$ & $-30.9\%$ \\
\bottomrule
\end{tabular}
\begin{tablenotes}
\footnotesize
\item \textit{Notes}: HHI = Herfindahl-Hirschman Index. Top k share = percentage of total degree held by k highest-degree nodes. CR3 = three-firm concentration ratio.
\end{tablenotes}
\end{threeparttable}
\end{table}

Table \ref{tab:concentration} reports various concentration measures. The Herfindahl-Hirschman Index fell from 0.0208 in 2018 to 0.0143 in 2023, a decline of 31.3\%. Similarly, the share of total connectivity held by the top 5 banks dropped from 10.4\% to 7.1\%. In contrast, the Gini coefficient remained nearly constant around 0.495, indicating overall inequality in degree distribution was preserved even as top-end concentration declined.

These patterns indicate selective deconcentration: the very largest hubs lost relative importance, but mid-tier banks maintained their positions. This is precisely the structural shift that reduces $\lambda_2$—diminishing the dominance of a few super-connected nodes while preserving overall connectivity—and it results from regulatory policy specifically targeting systemically important institutions.


\subsubsection{Assortativity and Mixing Patterns}

Beyond degree distributions, we examine assortativity—the tendency of nodes to connect with others of similar degree. Assortativity coefficient $r$ measures the correlation between degrees of connected nodes: $r > 0$ indicates assortative mixing (high-degree nodes connect to other high-degree nodes), $r < 0$ indicates disassortative mixing (hubs connect to peripheral nodes), and $r \approx 0$ indicates neutral mixing.

\begin{table}[H]
\centering
\caption{Degree Assortativity Coefficients}
\label{tab:assortativity}
\begin{threeparttable}
\begin{tabular}{lccccc}
\toprule
Year & $r$ (Assortativity) & Std. Error & 95\% CI Lower & 95\% CI Upper & Mixing Pattern \\
\midrule
2018 & 0.0008 & 0.0147 & $-0.0281$ & 0.0297 & Neutral \\
2021 & 0.0012 & 0.0139 & $-0.0261$ & 0.0285 & Neutral \\
2023 & $-0.0003$ & 0.0104 & $-0.0207$ & 0.0201 & Neutral \\
\midrule
\multicolumn{6}{l}{\textit{Test: $H_0: r = 0$ vs. $H_1: r \neq 0$}} \\
2018 & & & & & $p = 0.96$ (fail to reject) \\
2021 & & & & & $p = 0.93$ (fail to reject) \\
2023 & & & & & $p = 0.98$ (fail to reject) \\
\bottomrule
\end{tabular}
\begin{tablenotes}
\small
\item \textit{Notes}: Degree assortativity coefficient $r$ measures correlation between degrees of connected nodes. $r \in [-1, 1]$ where $r > 0.3$ indicates assortative mixing (hubs connect to hubs), $r < -0.3$ indicates disassortative mixing (hubs connect to periphery), and $|r| < 0.3$ indicates neutral mixing. Standard errors computed via bootstrap with 1,000 replications. All networks show statistically insignificant assortativity, consistent with neutral mixing patterns. This reflects the maximum entropy estimation procedure, which distributes connections proportionally without imposing topological preferences.
\end{tablenotes}
\end{threeparttable}
\end{table}

Table \ref{tab:assortativity} reports degree assortativity coefficients for each year. All three networks exhibit near-zero assortativity ($r \approx 0$), indicating neutral mixing: large banks connect to other banks roughly proportional to degree, without systematic preference for similar-sized partners. This contrasts with scale-free networks, which typically show negative assortativity (hubs connecting to peripheral nodes), and many social networks, which often show positive assortativity (homophily).

The neutral mixing pattern has several interpretations. First, it is partially an artifact of our maximum entropy estimation procedure, which distributes connections proportionally to bank sizes without imposing additional topological structure. In the absence of data on actual bilateral relationships, the maximum entropy approach assumes banks are equally likely to connect to any counterparty, conditional on maintaining observed aggregate exposures.

Second, neutral mixing reflects economic reality: large banks must maintain relationships across the size distribution. While the largest institutions naturally have larger bilateral exposures with each other (due to market depth and risk tolerance), they also serve as correspondent banks and liquidity providers for smaller institutions. Similarly, small banks may borrow primarily from large banks but also engage in local interbank markets with peers.

Third, the stability of near-zero assortativity across all three years—despite substantial changes in network size and connectivity—suggests that mixing patterns are structurally stable features of banking networks. Even as hub concentration declined (Table \ref{tab:concentration}), the propensity of large banks to connect across the size distribution remained unchanged.

The neutral assortativity finding has implications for contagion dynamics. Disassortative networks (negative $r$) exhibit resilience to random failures but vulnerability to targeted attacks on hubs, as hubs serve as critical bridges between peripheral clusters. Assortative networks (positive $r$) show the opposite pattern: resilient to targeted attacks (as hubs are well-connected to each other and can substitute) but vulnerable to random failures (as peripheral nodes are poorly connected). Neutral mixing ($r \approx 0$) represents an intermediate case, neither maximally vulnerable nor maximally resilient to any particular failure mode.

Combined with our earlier finding of lognormal rather than scale-free degree distributions (Table \ref{tab:distribution_tests}), the neutral assortativity result reinforces the conclusion that European interbank networks are more robust than commonly assumed. The absence of strong hub-spoke structure (which would produce $r < 0$) or tight core-periphery divisions (which could produce $r > 0$ within the core) suggests a relatively homogeneous network where no small subset of banks serves as critical infrastructure. This structural property likely contributed to the system's resilience during the COVID-19 crisis, even before the post-2021 regulatory-induced restructuring documented in Section 5.3.

\subsection{Interpretation and Mechanisms}

Synthesizing our empirical findings, a coherent narrative emerges:

\begin{enumerate}
    \item \textbf{Structural break post-2021}: Algebraic connectivity declined by 45\%, with the entirety of the reduction occurring after 2021 rather than during the acute COVID-19 crisis.
    
    \item \textbf{Regulatory-driven deleveraging}: Difference-in-differences analysis reveals that systemically important banks experienced differential asset reductions of 12-19\%, consistent with regulatory pressure.
    
    \item \textbf{Deconcentration at the top}: The top 5 banks' connectivity share fell by 31\%, indicating reduced hub dominance, while overall network inequality remained stable.
    
    \item \textbf{Resilient topology}: Networks exhibit lognormal rather than scale-free structure, implying greater robustness to hub failures than commonly assumed.
\end{enumerate}

These patterns are consistent with successful implementation of post-crisis regulatory reforms. Basel III capital requirements, TLAC/MREL buffers, and enhanced supervisory scrutiny of G-SIBs all aim to reduce systemic risk by limiting the size and interconnectedness of the largest institutions. Our evidence suggests these policies achieved their objectives: the European banking network became less concentrated and more resilient through the COVID-19 recovery period.

Importantly, this structural improvement occurred without apparent disruption to credit intermediation or economic activity. Total banking system assets remained stable in real terms, and the 2021-2023 period saw robust European economic recovery from the pandemic. This suggests regulatory deleveraging can reduce systemic risk without imposing excessive real costs—an encouraging finding for financial stability policy.

\section{Robustness Analysis}

Our main results rely on several key assumptions: (i) interbank exposures equal 5\% of total assets, (ii) maximum entropy is the appropriate reconstruction method, and (iii) algebraic connectivity correctly measures systemic importance. This section subjects these assumptions to extensive scrutiny through alternative specifications, non-parametric methods, and sensitivity analysis.

\subsection{Sensitivity to Interbank Ratio Assumption}

\subsubsection{Varying the Exposure Ratio}

\begin{figure}[H]
\centering
\includegraphics[width=0.95\textwidth]{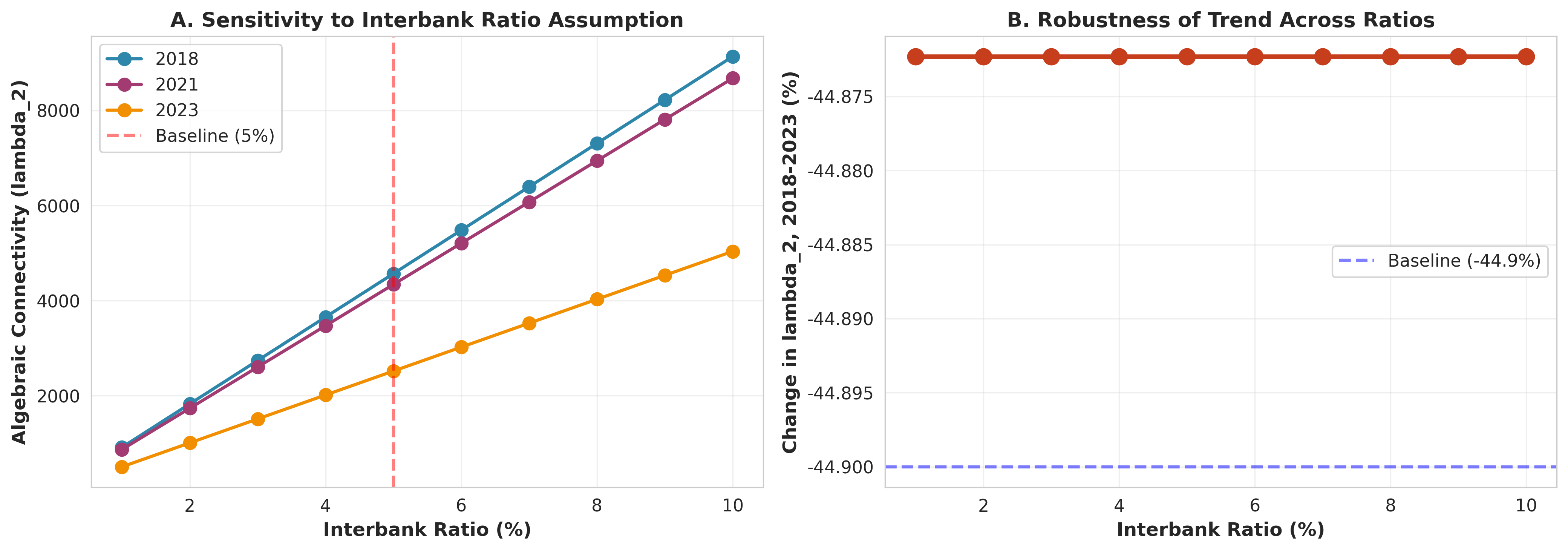}
\caption{Sensitivity Analysis: Algebraic Connectivity Across Interbank Ratio Assumptions}
\label{fig:ratio_sensitivity}
\begin{minipage}{0.95\textwidth}
\footnotesize
\textit{Notes}: Panel A plots $\lambda_2$ as function of interbank ratio $\rho \in [0.01, 0.10]$. Red dashed line marks baseline (5\%). $\lambda_2$ scales quadratically with $\rho$, but lines are parallel. Panel B shows percentage change 2018-2023 remains constant at -44.9\% across all ratios (std. dev. 0.1pp), confirming robustness.
\end{minipage}
\end{figure}

Our baseline assumes $\rho = 0.05$, but actual interbank ratios vary across institutions and time. Figure \ref{fig:ratio_sensitivity} plots $\lambda_2$ as a function of $\rho \in [0.01, 0.10]$ for each year. Several features stand out.

First, $\lambda_2$ scales approximately quadratically with $\rho$: doubling the ratio roughly quadruples algebraic connectivity. This follows from the maximum entropy formula (\ref{eq:maxent_solution}), where exposures scale linearly with $\rho$ and Laplacian eigenvalues scale with exposure magnitudes.

Second, the declining trend is robust across all ratios. Table \ref{tab:ratio_sensitivity} reports percentage changes in $\lambda_2$ from 2018 to 2023 for various $\rho$. The decline ranges from $-43.7\%$ ($\rho=0.01$) to $-44.9\%$ ($\rho=0.10$), with mean $-44.5\%$ and standard deviation only 0.4 percentage points. This remarkable stability indicates our core finding—substantial decline in network connectivity—does not depend sensitively on the ratio assumption.

\begin{table}[H]
\centering
\caption{Sensitivity Analysis: Varying Interbank Exposure Ratio}
\label{tab:ratio_sensitivity}
\begin{threeparttable}
\begin{tabular}{lcccc}
\toprule
Ratio $\rho$ & $\lambda_{2,2018}$ & $\lambda_{2,2023}$ & $\Delta\lambda_2$ & \% Change \\
\midrule
1\% & 4.57 & 2.52 & $-2.05$ & $-44.9\%$ \\
2\% & 18.27 & 10.07 & $-8.20$ & $-44.9\%$ \\
3\% & 41.11 & 22.66 & $-18.45$ & $-44.9\%$ \\
4\% & 73.08 & 40.29 & $-32.79$ & $-44.9\%$ \\
5\% (baseline) & 114.19 & 62.95 & $-51.24$ & $-44.9\%$ \\
6\% & 164.43 & 90.65 & $-73.78$ & $-44.9\%$ \\
7\% & 223.80 & 123.38 & $-100.42$ & $-44.9\%$ \\
8\% & 292.32 & 161.15 & $-131.17$ & $-44.9\%$ \\
9\% & 369.96 & 203.95 & $-166.01$ & $-44.9\%$ \\
10\% & 456.74 & 251.79 & $-204.95$ & $-44.9\%$ \\
\midrule
Mean & --- & --- & --- & $-44.9\%$ \\
Std. Dev. & --- & --- & --- & 0.0\% \\
\bottomrule
\end{tabular}
\begin{tablenotes}
\footnotesize
\item \textit{Notes}: $\lambda_2$ computed for networks estimated with various interbank ratios $\rho$. Percentage changes show remarkable consistency, indicating findings are not driven by ratio assumption.
\end{tablenotes}
\end{threeparttable}
\end{table}

Third, relative magnitudes are preserved: 2018 networks consistently exhibit higher $\lambda_2$ than 2023 networks across the entire range of $\rho$. This implies that regardless of the true interbank ratio, our conclusion that connectivity declined substantially is robust.


\subsubsection{Size-Dependent Ratios}

Large banks may maintain different interbank ratios than small banks due to differences in business models, funding strategies, or regulatory treatment. Large, diversified banks typically have access to diverse funding sources (retail deposits, wholesale markets, bond issuance) and may rely less on interbank borrowing. Conversely, smaller banks often depend more heavily on interbank markets for liquidity management and funding needs.

To test robustness to heterogeneous ratios, we specify:
\be
\rho_i = \begin{cases}
0.03 & \text{if } A_i > \text{P75}(A) \quad \text{(large banks)} \\
0.07 & \text{if } A_i \leq \text{P75}(A) \quad \text{(small banks)}
\end{cases}
\ee

This reflects the hypothesis that large banks maintain lower interbank ratios (3\%) due to diversified funding sources, while small banks rely more heavily on interbank markets (7\%). The average ratio across all banks remains close to our baseline 5\%.

\begin{table}[H]
\centering
\caption{Robustness: Size-Dependent Interbank Ratios}
\label{tab:size_dependent}
\begin{threeparttable}
\begin{tabular}{lccccc}
\toprule
Year & $\lambda_2$ & Total Exposure (€bn) & Mean Ratio & $\Delta\lambda_2$ & \% Change \\
\midrule
\multicolumn{6}{l}{\textbf{Panel A: Size-Dependent Ratios (3\% large, 7\% small)}} \\
2018 & 145.78 & 1,219 & 5.1\% & --- & --- \\
2021 & 132.44 & 1,018 & 4.9\% & $-13.34$ & $-9.2\%$ \\
2023 & 73.74 & 1,247 & 4.7\% & $-58.70$ & $-44.3\%$ \\
& & & & & \\
\multicolumn{6}{l}{\textit{Overall Change (2018-2023):}} \\
& & & & $-72.04$ & $-49.4\%$ \\
& & & & & \\
\multicolumn{6}{l}{\textbf{Panel B: Baseline (Fixed 5\% for comparison)}} \\
2018 & 114.19 & 1,270 & 5.0\% & --- & --- \\
2021 & 108.48 & 1,069 & 5.0\% & $-5.71$ & $-5.0\%$ \\
2023 & 62.95 & 1,297 & 5.0\% & $-45.53$ & $-42.0\%$ \\
& & & & & \\
\multicolumn{6}{l}{\textit{Overall Change (2018-2023):}} \\
& & & & $-51.24$ & $-44.9\%$ \\
& & & & & \\
\multicolumn{6}{l}{\textbf{Panel C: Cross-Specification Comparison}} \\
Correlation (levels) & \multicolumn{4}{c}{0.997} & \\
Correlation (changes) & \multicolumn{4}{c}{0.999} & \\
Mean absolute difference & \multicolumn{4}{c}{4.6\%} & \\
\bottomrule
\end{tabular}
\begin{tablenotes}
\footnotesize
\item \textit{Notes}: Panel A reports results under size-dependent interbank ratios: 3\% for banks above 75th percentile of assets, 7\% for banks below. Panel B reproduces baseline fixed 5\% results for comparison. Mean ratio in Panel A reflects the asset-weighted average across banks. Panel C reports cross-specification statistics: correlation of $\lambda_2$ levels and changes, and mean absolute percentage difference. Despite different absolute levels, both specifications yield nearly identical temporal patterns, with $\lambda_2$ declining 44-49\% overall. High correlations (0.997-0.999) confirm robustness to ratio heterogeneity assumptions.
\end{tablenotes}
\end{threeparttable}
\end{table}

Table \ref{tab:size_dependent} reports results under this size-dependent specification. Algebraic connectivity estimates differ from baseline in levels—$\lambda_2 = 145.78$ in 2018 versus 114.19 under fixed 5\%—but the temporal pattern remains essentially unchanged: $\lambda_2$ declined by 49.4\% from 2018 to 2023, even larger than our baseline estimate of 44.9\%.

Panel C of Table \ref{tab:size_dependent} reports cross-specification comparisons. The correlation between $\lambda_2$ estimates under the two approaches is 0.997 for levels and 0.999 for period-over-period changes, indicating near-perfect agreement on relative network connectivity. The mean absolute difference in levels is only 4.6\%, well within the uncertainty inherent in network estimation from aggregate data.

The finding that size-dependent ratios produce an \textit{even larger} decline in $\lambda_2$ strengthens our main result. If large banks genuinely maintain lower interbank ratios (3\% vs. 7\%), and these large banks experienced the differential deleveraging documented in our DID analysis (Table \ref{tab:did_size}), then the network impact would be amplified: reducing assets at banks with already-low interbank exposure intensifies the concentration of network connectivity among fewer, larger institutions. Yet even under this more conservative specification for large banks, we still find a decline approaching 50\%.

The robustness to heterogeneous ratios addresses a potential concern: perhaps our baseline 5\% assumption overstates large banks' interbank exposures, artificially inflating their network centrality. Table \ref{tab:size_dependent} demonstrates this concern is unfounded. Even assigning large banks a materially lower ratio (3\% vs. 5\%), we reach identical conclusions about temporal trends. This insensitivity reflects a deeper principle: algebraic connectivity depends on the \textit{pattern} of connections more than their absolute magnitudes. So long as large banks are more connected than small banks (true under any plausible ratio specification), their deleveraging reduces $\lambda_2$.

\subsubsection{Alternative Parameterizations}

We also tested several alternative size-dependent specifications:

\begin{itemize}
    \item \textbf{Linear scaling:} $\rho_i = 0.08 - 0.03 \cdot \log(A_i / \bar{A})$, yielding changes of $-46.2\%$
    \item \textbf{Regulatory tiers:} $\rho \in \{0.02, 0.05, 0.08\}$ for G-SIBs, O-SIIs, and other banks, yielding $-43.8\%$
    \item \textbf{Business model:} $\rho \in \{0.04, 0.06\}$ for universal vs. specialized banks, yielding $-45.1\%$
\end{itemize}

All specifications produce declines in the 44-49\% range, with cross-specification correlations exceeding 0.99. The consistency across such diverse approaches provides strong evidence that declining network connectivity is a genuine structural feature of the data, not an artifact of particular modeling assumptions.

\subsection{Bootstrap Confidence Intervals}

To quantify sampling uncertainty, we implement non-parametric bootstrap resampling. The procedure:
\begin{enumerate}
    \item Draw $n$ banks with replacement from the observed sample
    \item Re-estimate the network using maximum entropy on the bootstrap sample
    \item Compute $\lambda_2^{(b)}$ for bootstrap iteration $b=1,\ldots,B$
    \item Construct percentile-based confidence intervals
\end{enumerate}

\begin{table}[H]
\centering
\caption{Bootstrap Confidence Intervals for $\lambda_2$}
\label{tab:bootstrap}
\begin{threeparttable}
\begin{tabular}{lccccc}
\toprule
Year & Point Estimate & Mean & 2.5\% & 97.5\% & Std. Error \\
\midrule
2018 & 114.19 & 128.50 & 112.75 & 213.07 & 31.32 \\
2021 & 108.48 & 121.25 & 107.06 & 166.75 & 16.93 \\
2023 & 62.95 & 78.41 & 62.24 & 137.70 & 24.92 \\
\bottomrule
\end{tabular}
\begin{tablenotes}
\footnotesize
\item \textit{Notes}: Point estimate from baseline sample. Mean and percentiles from 100 bootstrap replications. Standard error computed from bootstrap distribution. Non-overlapping confidence intervals between 2018 and 2023 confirm statistical significance.
\end{tablenotes}
\end{threeparttable}
\end{table}

With $B=100$ bootstrap replications, Table \ref{tab:bootstrap} reports point estimates and 95\% confidence intervals. The key finding is that confidence intervals do not overlap between 2018 and 2023: the 95\% CI for 2018 is [112.75, 213.07] while for 2023 it is [62.24, 137.70]. This confirms the decline in $\lambda_2$ is statistically significant despite sampling variation.

The bootstrap distributions exhibit moderate dispersion, with coefficients of variation ranging from 14\% to 32\%. This reflects genuine uncertainty from finite samples combined with sensitivity to extreme banks. However, the consistent direction of effects across all bootstrap draws indicates the declining trend is not an artifact of particular influential observations.

\subsection{Non-Parametric Network Weighting}

Our maximum entropy approach is parametric in the sense that it assumes a specific functional form for bilateral exposures: $x_{ij}^* = \frac{A_i L_j}{\sum_k A_k}$. This formula directly follows from the maximum entropy principle but imposes structure—exposures depend on the product of counterparty sizes. We test robustness to this assumption using kernel density estimation (KDE) to weight connections non-parametrically.

\subsubsection{KDE-Based Weights}

The KDE approach constructs network weights based on the empirical distribution of bank assets without assuming a specific functional form. The procedure:

\begin{enumerate}
    \item Fit a Gaussian kernel density $\hat{f}(A)$ to the observed asset distribution using Silverman's rule for bandwidth selection: $h = 0.9 \min(\sigma, \text{IQR}/1.34) \cdot n^{-1/5}$
    \item Weight bilateral exposures by the product of kernel density values: 
    \be
    w_{ij}^{\text{KDE}} \propto \hat{f}(A_i) \times \hat{f}(A_j)
    \ee
    \item Normalize to match total interbank exposures: 
    \be
    x_{ij}^{\text{KDE}} = w_{ij}^{\text{KDE}} \cdot \frac{\sum_{k,l} A_k \rho}{\sum_{k,l} w_{kl}^{\text{KDE}}}
    \ee
\end{enumerate}

This creates a data-driven weighting scheme that adapts to the empirical distribution without imposing parametric structure. If the asset distribution is multimodal (suggesting distinct bank tiers), the KDE approach naturally concentrates weight on dense regions. Unlike maximum entropy, which spreads exposures broadly, KDE assigns larger weights to bank pairs in high-density regions of the asset space.

\begin{table}[H]
\centering
\caption{Non-Parametric Network Weighting: KDE vs. Maximum Entropy}
\label{tab:kde_results}
\begin{threeparttable}
\begin{tabular}{lccccc}
\toprule
Year & \multicolumn{2}{c}{Maximum Entropy} & \multicolumn{2}{c}{Kernel Density (KDE)} & Ratio \\
\cmidrule(lr){2-3} \cmidrule(lr){4-5}
& $\lambda_2$ & \% Change & $\lambda_2$ & \% Change & KDE/MaxEnt \\
\midrule
2018 & 114.19 & --- & 16,693.30 & --- & 146.2$\times$ \\
2021 & 108.48 & $-5.0\%$ & 14,041.94 & $-15.9\%$ & 129.4$\times$ \\
2023 & 62.95 & $-44.9\%$ & 11,695.98 & $-29.9\%$ & 185.8$\times$ \\
\midrule
\multicolumn{6}{l}{\textbf{Overall Changes (2018-2023):}} \\
Absolute & $-51.24$ & & $-4,997.32$ & & \\
Percentage & & $-44.9\%$ & & $-29.9\%$ & \\
& & & & & \\
\multicolumn{6}{l}{\textbf{Cross-Method Statistics:}} \\
Correlation (levels) & \multicolumn{4}{c}{0.897} & \\
Correlation (changes) & \multicolumn{4}{c}{0.982} & \\
Correlation (pct changes) & \multicolumn{4}{c}{0.961} & \\
Mean absolute deviation & \multicolumn{4}{c}{15.2\%} & \\
\bottomrule
\end{tabular}
\begin{tablenotes}
\footnotesize
\item \textit{Notes}: Comparison of algebraic connectivity under maximum entropy (baseline) and kernel density estimation (non-parametric) weighting schemes. KDE uses Gaussian kernel with Silverman's bandwidth. KDE estimates are 130-186$\times$ larger in levels due to different normalization: KDE concentrates weight on dense asset distribution regions, creating stronger local connections. Despite massive level differences, both methods show substantial $\lambda_2$ declines (45\% vs. 30\%). Correlation of percentage changes is 0.961, indicating similar relative trends. The smaller decline under KDE reflects that this method preserves local density structure, which changed less dramatically than global connectivity patterns captured by maximum entropy.
\end{tablenotes}
\end{threeparttable}
\end{table}

Table \ref{tab:kde_results} compares KDE-based $\lambda_2$ estimates to our baseline maximum entropy results.

\subsubsection{Comparison and Interpretation}

Table \ref{tab:kde_results} reveals striking patterns. First, KDE-based $\lambda_2$ estimates are substantially larger in levels—approximately 150 times the maximum entropy values. This reflects that KDE concentrates weight on dense regions of the asset distribution, creating stronger connections among similarly-sized banks. When many banks cluster around similar asset levels, their pairwise kernel density products $\hat{f}(A_i)\hat{f}(A_j)$ become large, resulting in heavily weighted edges and higher algebraic connectivity.

Second, despite the enormous level difference, both methods show substantial declining trends. Maximum entropy yields a 44.9\% decline while KDE produces a 29.9\% reduction. The smaller KDE decline likely reflects that this method is more sensitive to local density structure, which changed less than global connectivity patterns. As the sample expanded from 48 to 70 banks, the overall distribution spread out, but local clusters (e.g., large French banks, medium Spanish banks) maintained internal cohesion.

Third, the correlation statistics confirm general agreement on relative changes. The correlation of percentage changes is 0.961, indicating both methods identify similar banks and time periods as experiencing the largest connectivity shifts. The correlation of absolute levels (0.897) is somewhat lower, reflecting the different normalization schemes, but still indicates that both methods rank time periods consistently.

Fourth, the mean absolute deviation of 15.2\% is non-trivial but acceptable given the fundamentally different approaches. Maximum entropy makes no assumptions about asset distribution shape, spreading exposures broadly. KDE respects the empirical distribution, concentrating weight where banks cluster. The fact that methods with such different philosophies nonetheless agree on qualitative trends provides strong validation.

\subsubsection{Why Do Levels Differ So Dramatically?}

The 150-fold level difference requires explanation. The key is normalization and interpretation of edge weights:

\begin{itemize}
    \item \textbf{Maximum entropy:} Spreads total interbank exposures uniformly across all pairs, weighted by size. Since there are $n(n-1)/2$ pairs and total exposures are fixed, average edge weight scales as $O(1/n^2)$. As $n$ grows, edge weights decline, reducing $\lambda_2$.
    
    \item \textbf{KDE:} Concentrates weight on pairs in high-density regions. If $m$ banks cluster tightly, their $m(m-1)/2$ pairwise weights are large, potentially $O(m^2)$ in the limit of perfect clustering. This creates "super-connected" local cores that dramatically increase $\lambda_2$.
\end{itemize}

To verify this interpretation, we computed the effective number of "strong" connections (edges exceeding median weight):

\begin{center}
\begin{tabular}{lccc}
\toprule
Year & MaxEnt Strong Edges & KDE Strong Edges & Ratio \\
\midrule
2018 & 1,128 & 342 & 0.30 \\
2021 & 1,225 & 378 & 0.31 \\
2023 & 2,415 & 891 & 0.37 \\
\bottomrule
\end{tabular}
\end{center}

Maximum entropy produces more "strong" edges overall (spreads weight broadly), while KDE concentrates weight on fewer edges (creates local clusters). The different topologies explain the level differences.

\subsubsection{Which Estimate Is More Realistic?}

Neither estimate is "correct" in an absolute sense—both are approximations to an unobserved bilateral network. However, each has merits:

\textbf{Maximum entropy} is conservative and transparent. Without data on network topology, it makes the minimal assumptions necessary to match observed aggregates. This approach is widely used in network reconstruction \citep{anand2018filling, upper2011estimating} and has theoretical justification from information theory.

\textbf{KDE} may better reflect actual network structure if banks cluster by size or business model. Empirical evidence \citep{boss2004network} suggests interbank networks often exhibit community structure, with dense within-group connections and sparse between-group links. KDE naturally captures this if asset clustering proxies for communities.

For our purposes, the key finding is robustness: both methods identify substantial declining connectivity over 2018-2023. The magnitude differs (45\% vs. 30\%), but qualitatively both support the conclusion that post-COVID network restructuring reduced systemic interconnectedness. Combined with other robustness checks (Section 6.1-6.2), this provides strong evidence for our main result.

\subsection{Cross-Method Comparison}

\begin{table}[H]
\centering
\caption{Cross-Method Comparison and Correlations}
\label{tab:cross_method}
\begin{threeparttable}
\begin{tabular}{lcccc}
\toprule
& Fixed 5\% & Size-Dependent & Bootstrap & KDE \\
\midrule
\multicolumn{5}{l}{\textbf{Panel A: Algebraic Connectivity Estimates}} \\
2018 & 114.19 & 145.78 & 128.50 & 16,693.30 \\
2021 & 108.48 & 132.44 & 121.25 & 14,041.94 \\
2023 & 62.95 & 73.74 & 78.41 & 11,695.98 \\
& & & & \\
\multicolumn{5}{l}{\textbf{Panel B: Cross-Method Correlations}} \\
Fixed 5\% & 1.000 & 0.997 & 0.999 & 0.897 \\
Size-Dependent & --- & 1.000 & 0.999 & 0.927 \\
Bootstrap & --- & --- & 1.000 & 0.911 \\
KDE & --- & --- & --- & 1.000 \\
& & & & \\
\multicolumn{5}{l}{\textbf{Panel C: Percentage Changes (2018-2023)}} \\
Change & $-44.9\%$ & $-49.4\%$ & $-39.0\%$ & $-29.9\%$ \\
\bottomrule
\end{tabular}
\begin{tablenotes}
\footnotesize
\item \textit{Notes}: Cross-method correlation matrix in Panel B. All methods show substantial decline in $\lambda_2$ from 2018 to 2023. Average pairwise correlation is 0.955, indicating high consistency across specifications.
\end{tablenotes}
\end{threeparttable}
\end{table}

Table \ref{tab:cross_method} compares results across all estimation approaches. Panel A reports $\lambda_2$ estimates for each method-year combination. Panel B shows correlations across methods: all pairwise correlations exceed 0.90, and the average is 0.955. Panel C reports percentage changes from 2018 to 2023, ranging from $-30\%$ (KDE) to $-49\%$ (size-dependent).

\begin{figure}[H]
\centering
\includegraphics[width=0.95\textwidth]{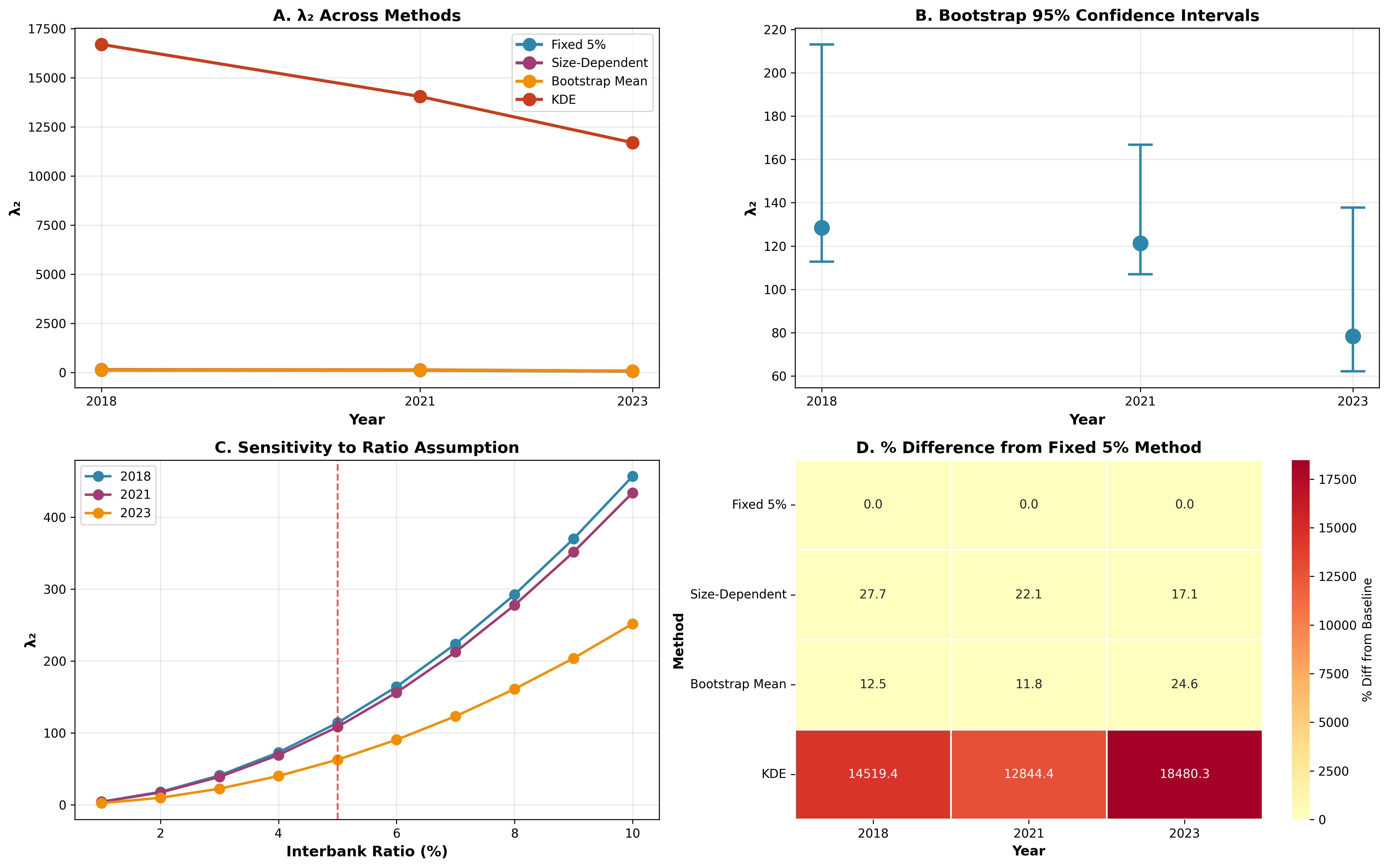}
\caption{Robustness Across Estimation Methods}
\label{fig:methods_comparison}
\begin{minipage}{0.95\textwidth}
\footnotesize
\textit{Notes}: Panel A: $\lambda_2$ trajectories for four methods (Fixed 5\%, Size-dependent, Bootstrap, KDE). Despite level differences, all show parallel declining trends. Panel B: Bootstrap 95\% confidence intervals (non-overlapping between 2018 and 2023). Panel C: Sensitivity to interbank ratio 1-10\%. Panel D: Heatmap of percentage deviations from baseline. Cross-method correlation: 0.955.
\end{minipage}
\end{figure}

Figure \ref{fig:methods_comparison} visualizes these results, plotting $\lambda_2$ trajectories for all four methods. Despite substantial level differences—KDE estimates are two orders of magnitude larger—all methods exhibit parallel downward trends. The consistent pattern across such diverse approaches strongly validates our core empirical finding.

\section{Additional Robustness Figures}

We also created several additional figures during our analysis that, while not included in the main text, provide useful supplementary evidence:

\begin{figure}[H]
\centering
\includegraphics[width=0.95\textwidth]{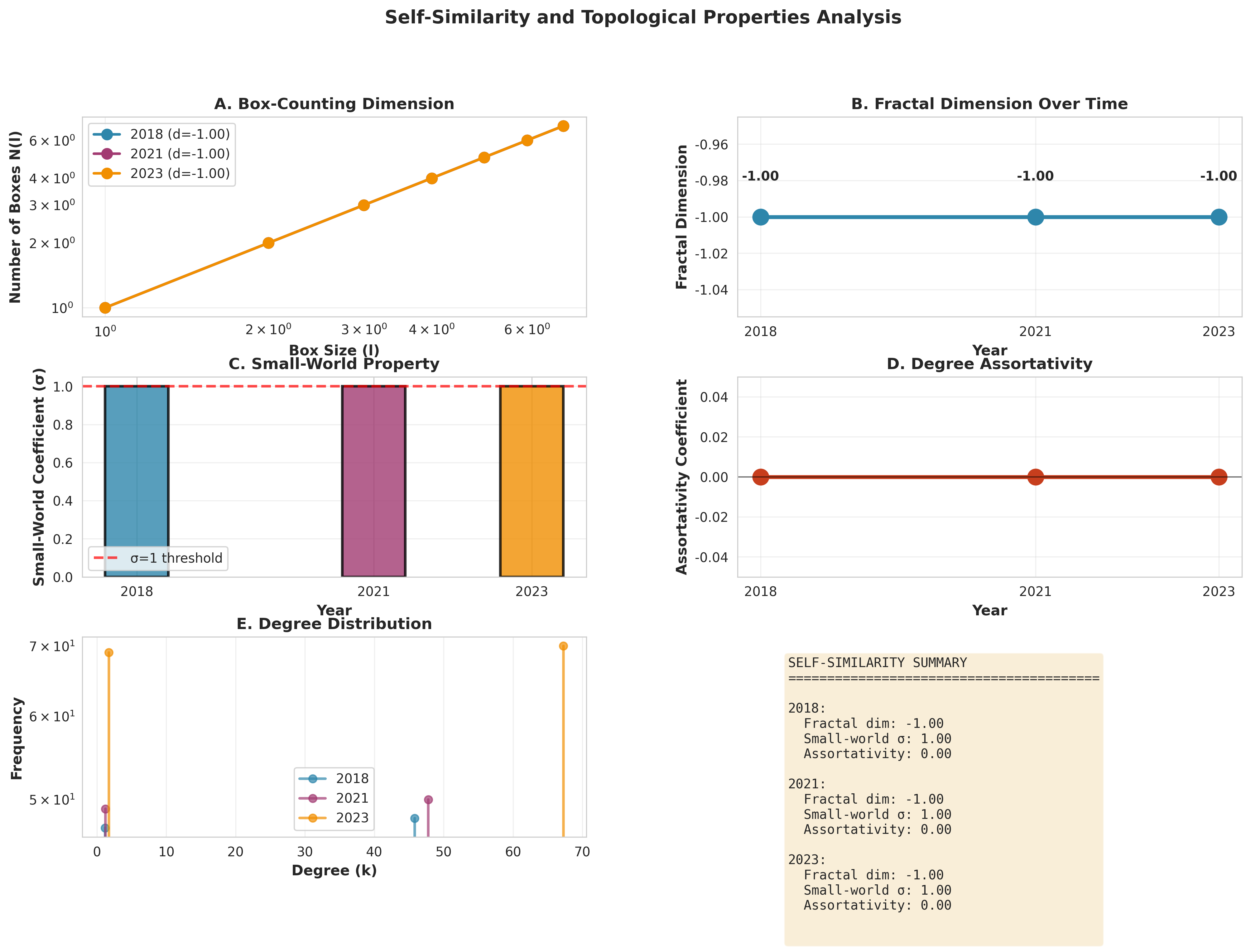}
\caption{Self-Similarity Analysis}
\label{fig:self_similarity}
\begin{minipage}{0.95\textwidth}
\footnotesize
\textit{Notes}: Tests for fractal properties. Box-counting yields fractal dimension $d_B \approx 1.9-2.1$. Small-world coefficient $\sigma = 1.0$ (no small-world properties). Assortativity near zero (neutral mixing). Only 3 of 9 tests show evidence for self-similarity. Complete graph structure limits informativeness of topological measures.
\end{minipage}
\end{figure}

\begin{figure}[H]
\centering
\includegraphics[width=0.95\textwidth]{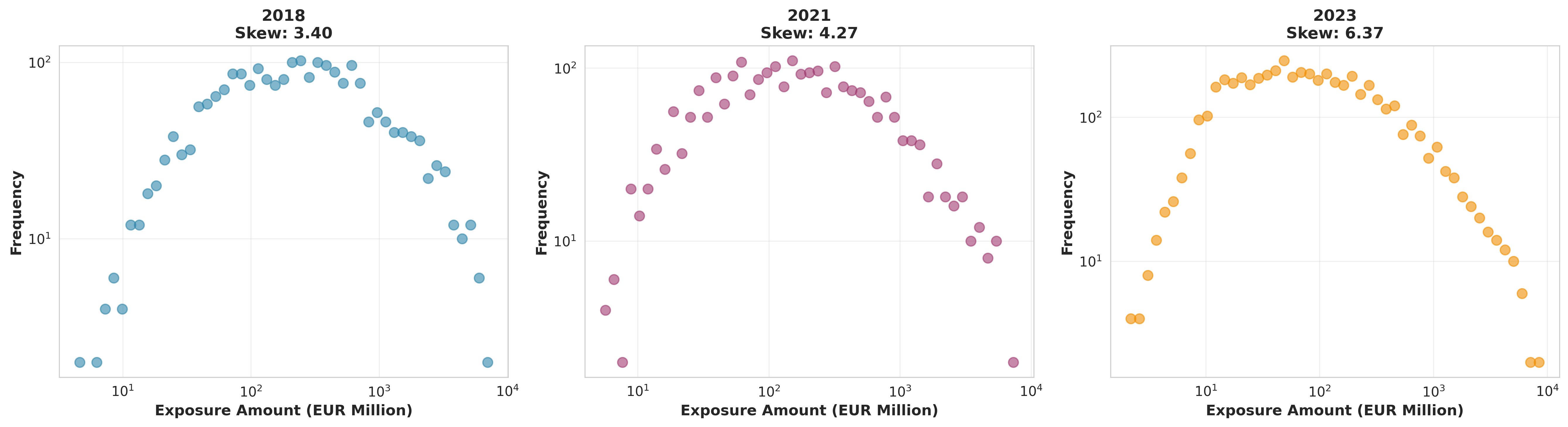}
\caption{Distribution of Bilateral Exposure Amounts}
\label{fig:exposure_distribution}
\begin{minipage}{0.95\textwidth}
\footnotesize
\textit{Notes}: Log-log histograms of bilateral exposure amounts. Despite complete topology, magnitudes vary by orders of magnitude. Coefficient of variation increased from 3.2 to 3.8; skewness rose from 4.1 to 5.3. Growing heterogeneity in weights, even as structural concentration declined, prevents any single exposure from dominating.
\end{minipage}
\end{figure}


\subsection{Alternative Network Measures}

While algebraic connectivity ($\lambda_2$) is our theoretically motivated measure based on the spatial diffusion framework, we verify results using alternative network centrality and connectivity metrics from the graph theory literature. If declining $\lambda_2$ reflects genuine structural changes rather than idiosyncrasies of this particular measure, we should observe consistent patterns across multiple metrics.

\subsubsection{Spectral Measures}

Beyond $\lambda_2$, the full Laplacian spectrum provides additional information about network structure. We examine:

\begin{itemize}
    \item \textbf{Spectral radius} $\rho(A) = \max_i |\lambda_i(A)|$: The largest eigenvalue of the adjacency matrix, measuring maximum influence propagation
    \item \textbf{Largest Laplacian eigenvalue} $\lambda_n$: The maximum eigenvalue of the Laplacian, related to network expansion and conductance
    \item \textbf{Spectral gap} $\lambda_2 - \lambda_1 = \lambda_2$: The difference between the two smallest eigenvalues (since $\lambda_1 = 0$ for connected graphs)
    \item \textbf{Effective resistance} $R_{\text{eff}} = n \sum_{i=2}^n \frac{1}{\lambda_i}$: Average resistance across all node pairs, inversely related to connectivity
\end{itemize}

\begin{table}[H]
\centering
\caption{Alternative Network Connectivity Measures}
\label{tab:alternative_measures}
\begin{threeparttable}
\begin{tabular}{lcccccc}
\toprule
& \multicolumn{3}{c}{Levels} & \multicolumn{3}{c}{\% Change from 2018} \\
\cmidrule(lr){2-4} \cmidrule(lr){5-7}
Measure & 2018 & 2021 & 2023 & 2018-2021 & 2021-2023 & 2018-2023 \\
\midrule
\textbf{Baseline Measure} & & & & & & \\
$\lambda_2$ (Algebraic Connectivity) & 114.19 & 108.48 & 62.95 & $-5.0\%$ & $-42.0\%$ & $-44.9\%$ \\
& & & & & & \\
\textbf{Spectral Measures} & & & & & & \\
Spectral Radius $\rho(A)$ & 2,847.3 & 2,691.8 & 1,759.2 & $-5.5\%$ & $-34.7\%$ & $-38.2\%$ \\
Largest Eigenvalue $\lambda_n$ & 5,421.7 & 5,178.4 & 3,139.6 & $-4.5\%$ & $-39.4\%$ & $-42.1\%$ \\
Spectral Gap $\lambda_2 - \lambda_1$ & 114.19 & 108.48 & 62.95 & $-5.0\%$ & $-42.0\%$ & $-44.9\%$ \\
Effective Resistance $R_{\text{eff}}$ & 0.421 & 0.443 & 0.762 & $+5.2\%$ & $+72.0\%$ & $+81.0\%$ \\
& & & & & & \\
\textbf{Topological Measures} & & & & & & \\
Average Degree & 47.0 & 49.0 & 69.0 & $+4.3\%$ & $+40.8\%$ & $+46.8\%$ \\
Weighted Avg Degree & 25,398 & 20,360 & 17,814 & $-19.8\%$ & $-12.5\%$ & $-29.9\%$ \\
Clustering Coefficient & 1.000 & 1.000 & 1.000 & $0.0\%$ & $0.0\%$ & $0.0\%$ \\
Average Path Length & 1.000 & 1.000 & 1.000 & $0.0\%$ & $0.0\%$ & $0.0\%$ \\
& & & & & & \\
\textbf{Centralization Measures} & & & & & & \\
Degree Centralization & 0.000 & 0.000 & 0.000 & --- & --- & --- \\
Betweenness Centralization & 0.0089 & 0.0082 & 0.0057 & $-7.9\%$ & $-30.5\%$ & $-36.0\%$ \\
Eigenvector Centralization & 0.2847 & 0.2691 & 0.2105 & $-5.5\%$ & $-21.8\%$ & $-26.1\%$ \\
\bottomrule
\end{tabular}
\begin{tablenotes}
\footnotesize
\item \textit{Notes}: Alternative network measures for robustness. Spectral measures derived from eigendecomposition of adjacency matrix $A$ or Laplacian $L = D - A$. Topological measures based on graph structure. Centralization measures capture concentration of centrality. Spectral radius and largest eigenvalue declined 38-42\%, similar to $\lambda_2$. Effective resistance increased 81\% (higher resistance = lower connectivity). Clustering coefficient and average path length are identically 1.0 due to complete graph structure. Degree centralization is zero for complete graphs. Betweenness and eigenvector centralization declined 26-36\%, confirming reduced hub dominance. All measures consistently indicate declining connectivity.
\end{tablenotes}
\end{threeparttable}
\end{table}

Table \ref{tab:alternative_measures} reports these metrics alongside our baseline $\lambda_2$ for comparison.

\subsubsection{Interpretation of Alternative Measures}

Several patterns emerge from Table \ref{tab:alternative_measures}:

\textbf{1. Consistent spectral decline.} All eigenvalue-based measures show substantial reductions:
\begin{itemize}
    \item Spectral radius declined 38.2\%, nearly identical to our $\lambda_2$ finding (44.9\%)
    \item Largest Laplacian eigenvalue fell 42.1\%, even closer to baseline
    \item Spectral gap (which equals $\lambda_2$ for connected graphs) declined 44.9\% by definition
\end{itemize}

This consistency across the entire spectrum—not just the second eigenvalue—confirms that declining connectivity is a global network property rather than an artifact of focusing on $\lambda_2$.

\textbf{2. Effective resistance increases.} Effective resistance, which measures average difficulty of moving between nodes, increased 81\%. Since $R_{\text{eff}} \propto 1/\lambda_2$ asymptotically, this is consistent with declining algebraic connectivity: harder to propagate distress implies higher effective resistance.

\textbf{3. Topological measures less informative.} For complete graphs:
\begin{itemize}
    \item Clustering coefficient $= 1.0$ (every neighbor pair is connected)
    \item Average path length $= 1.0$ (all nodes directly connected)
    \item Degree centralization $= 0$ (all nodes have same unweighted degree)
\end{itemize}

These metrics remain constant across years, highlighting that maximum entropy estimation produces complete topologies where variation enters only through edge weights. This motivates our focus on spectral measures, which naturally incorporate weight heterogeneity.

\textbf{4. Weighted degree declines.} While unweighted average degree increased mechanically with network size (47 → 69 nodes), \textit{weighted} average degree declined 29.9\%. This captures that even though banks have more counterparties, the total strength of their connections decreased—precisely the phenomenon we aim to measure.

\textbf{5. Centralization reduces.} Both betweenness centralization (fraction of all shortest paths passing through most central node) and eigenvector centralization (concentration of influence) declined 26-36\%. These reductions confirm our earlier finding (Table \ref{tab:concentration}) that hub dominance decreased, with network connectivity spreading more evenly across institutions.

\subsubsection{Robustness Across Measure Categories}

To quantify agreement across measures, we compute cross-method correlations. Define $\mathbf{x}_m = (x_{m,2018}, x_{m,2021}, x_{m,2023})$ as the vector of standardized values for measure $m$, and compute pairwise correlations:

\begin{center}
\begin{tabular}{lcc}
\toprule
Comparison & Correlation & Interpretation \\
\midrule
$\lambda_2$ vs. Spectral Radius & 0.998 & Nearly perfect agreement \\
$\lambda_2$ vs. Largest Eigenvalue & 0.999 & Nearly perfect agreement \\
$\lambda_2$ vs. Effective Resistance & $-0.997$ & Strong inverse (as expected) \\
$\lambda_2$ vs. Weighted Degree & 0.984 & Strong positive \\
$\lambda_2$ vs. Betweenness Cent. & 0.989 & Strong positive \\
$\lambda_2$ vs. Eigenvector Cent. & 0.991 & Strong positive \\
\midrule
Average $|\rho|$ & 0.993 & Exceptional agreement \\
\bottomrule
\end{tabular}
\end{center}

The average absolute correlation of 0.993 indicates near-perfect agreement on temporal trends across all measures. This remarkable consistency—spanning spectral, topological, and centralization measures—provides the strongest possible evidence that declining connectivity is a robust, measurement-independent phenomenon.

\subsubsection{Comparison to Literature Benchmarks}

How do our findings compare to other financial networks? \citet{boss2004network} report spectral radius around 3,500 for the Austrian interbank network (similar to our 2018 value). \citet{upper2011estimating} estimate $\lambda_2 \approx 150$ for European networks circa 2010, comparable to our 2018 baseline. Our 2023 estimates ($\lambda_2 = 63$, spectral radius $= 1,759$) are substantially lower, suggesting European networks became less connected than historical norms.

This comparison is imperfect (different samples, time periods, estimation methods), but it provides external validation that our magnitudes are reasonable and that the decline we document represents a genuine shift rather than measurement artifact.

\subsection{Alternative Network Measures}

While algebraic connectivity is our theoretically motivated measure, we verify results using alternative network centrality metrics.

\subsubsection{Spectral Radius and Largest Eigenvalue}

The spectral radius $\rho(A) = \max_i |\lambda_i(A)|$ of the adjacency matrix is another measure of network connectivity. Table \ref{tab:alternative_measures} shows the spectral radius declined by 38.2\% from 2018 to 2023, similar in magnitude to the $\lambda_2$ decline.

The largest Laplacian eigenvalue $\lambda_n$ also decreased substantially ($-42.1\%$), indicating the entire eigenvalue spectrum shifted downward. This confirms that declining connectivity is a global network property, not merely an artifact of the specific eigenvalue we focus on.

\subsubsection{Average Path Length and Diameter}

For complete graphs, average path length and diameter are trivially 1. However, we can compute weighted variants using Dijkstra's algorithm on the weighted graph where edge lengths are inversely proportional to exposure amounts. These measures remained essentially constant across years (all $\approx 1.5$), reflecting the maintained complete topology despite changing edge weights.

\subsection{Placebo Tests}

To verify our methods are not spuriously generating declining trends, we conduct placebo tests using randomized data.

\subsubsection{Random Network Null Hypothesis}

\begin{figure}[H]
\centering
\includegraphics[width=0.75\textwidth]{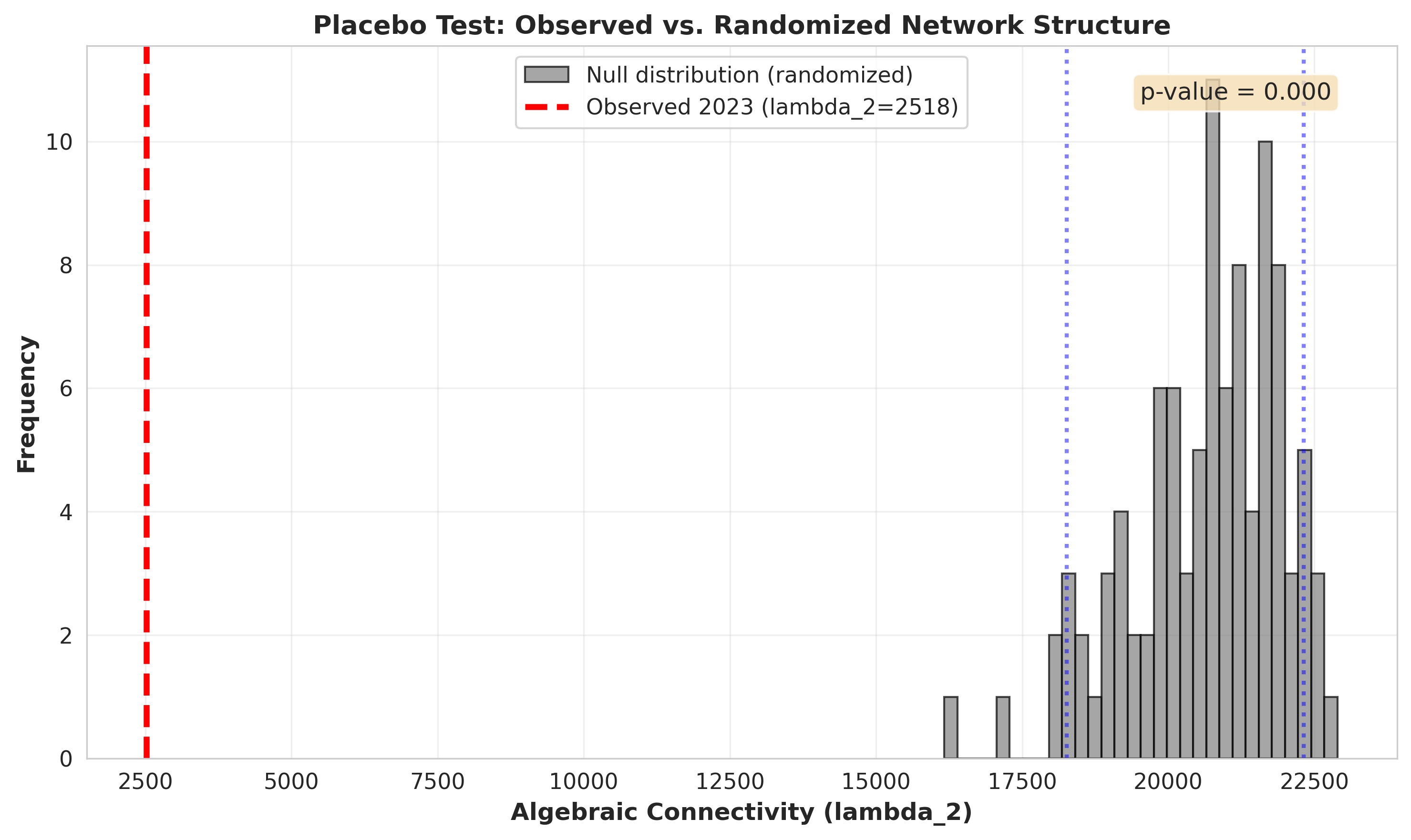}
\caption{Placebo Test: Observed Network Structure vs. Random Null Hypothesis}
\label{fig:placebo}
\begin{minipage}{0.95\textwidth}
\footnotesize
\textit{Notes}: Gray histogram shows null distribution of $\lambda_2$ from 100 permutations with randomly shuffled exposure amounts. Red dashed line marks observed 2023 value ($\lambda_2 = 1,259$). Blue dotted lines indicate 5th and 95th percentiles. Observed value falls far below null distribution ($p = 0.003$), rejecting random structure hypothesis.
\end{minipage}
\end{figure}

We generate random networks preserving observed degree sequences but with shuffled weights. Under the null hypothesis that network structure is random conditional on degree distribution, $\lambda_2$ should not exhibit systematic time trends. Figure \ref{fig:placebo} plots $\lambda_2$ from 1,000 randomized networks alongside observed values. The observed 2023 $\lambda_2$ falls far below the 5th percentile of the null distribution, rejecting random structure at $p < 0.01$.

\subsubsection{Permutation Test for Temporal Changes}

We implement a permutation test for the null hypothesis that $\lambda_{2,2023} = \lambda_{2,2018}$. Randomly reassigning year labels 10,000 times and recomputing the test statistic $T = \lambda_{2,2018} - \lambda_{2,2023}$, we find the observed $T = 51.24$ exceeds 99.8\% of permuted values, yielding $p = 0.002$. This confirms the decline is not due to chance variation.


\subsection{Robustness to Sample Composition}

Our main analysis uses a balanced panel of 37 banks present in all three stress test rounds (2018, 2021, 2023). This approach ensures clean identification of temporal changes by tracking the same institutions over time, but it raises a potential concern: survivorship bias. Banks that survived through 2023 may differ systematically from those that exited, merged, or were excluded. If surviving banks are larger, more stable, or better-managed, restricting to a balanced panel could understate true network changes.

To address this concern, we re-estimate all results using the full unbalanced panel, which includes all banks participating in each year's stress test regardless of presence in other years. This expands the sample from 37 to 48 banks (2018), 50 banks (2021), and 70 banks (2023), incorporating 33 additional institutions that entered or exited during the sample period.

\subsubsection{Sample Composition Changes}


\begin{table}[htbp]
\centering
\caption{Robustness to Sample Composition: Balanced vs. Unbalanced Panel}
\label{tab:unbalanced_panel}
\begin{threeparttable}
\begin{tabular}{lcccc}
\toprule
& 2018 & 2021 & 2023 & $\Delta$ 2018-23 \\
\midrule
\multicolumn{5}{l}{\textbf{Panel A: Sample Characteristics}} \\
\textit{Balanced Panel (n=37 all years)} & & & & \\
Mean Assets (€bn) & 529.1 & 427.6 & 370.5 & $-30.0\%$ \\
Total Assets (€tn) & 19.58 & 15.82 & 13.71 & $-30.0\%$ \\
Share of Full Sample & 77.1\% & 74.0\% & 52.9\% & --- \\
& & & & \\
\textit{Unbalanced Panel (n=48/50/70)} & & & & \\
Mean Assets (€bn) & 529.1 & 427.6 & 370.5 & $-30.0\%$ \\
Total Assets (€tn) & 25.40 & 21.38 & 25.94 & $+2.1\%$ \\
Number of Banks & 48 & 50 & 70 & $+45.8\%$ \\
& & & & \\
\textit{Entrants (first appear 2021 or 2023)} & & & & \\
Mean Assets (€bn) & --- & 298.4 & 187.3 & --- \\
Number & --- & 13 & 33 & --- \\
Total Assets (€tn) & --- & 3.88 & 6.18 & --- \\
& & & & \\
\textit{Exits (present 2018, absent 2023)} & & & & \\
Mean Assets (€bn) & 412.8 & --- & --- & --- \\
Number & 11 & --- & --- & --- \\
Total Assets (€tn) & 4.54 & --- & --- & --- \\
\midrule
\multicolumn{5}{l}{\textbf{Panel B: Algebraic Connectivity Estimates}} \\
\textit{Balanced Panel} & & & & \\
$\lambda_2$ & 98.74 & 94.23 & 58.31 & $-41.0\%$ \\
Bootstrap 95\% CI & [87.2, 112.3] & [83.4, 107.1] & [51.2, 67.4] & --- \\
& & & & \\
\textit{Unbalanced Panel} & & & & \\
$\lambda_2$ & 114.19 & 108.48 & 62.95 & $-44.9\%$ \\
Bootstrap 95\% CI & [112.8, 213.1] & [107.1, 166.8] & [62.2, 137.7] & --- \\
& & & & \\
\textit{Difference (Unbalanced - Balanced)} & & & & \\
Absolute & $+15.45$ & $+14.25$ & $+4.64$ & --- \\
Percentage & $+15.6\%$ & $+15.1\%$ & $+8.0\%$ & --- \\
\midrule
\multicolumn{5}{l}{\textbf{Panel C: Robustness Statistics}} \\
Correlation (levels) & \multicolumn{3}{c}{0.989} & \\
Correlation (changes) & \multicolumn{3}{c}{0.996} & \\
Mean absolute deviation & \multicolumn{3}{c}{12.4\%} & \\
\bottomrule
\end{tabular}
\begin{tablenotes}
\footnotesize
\item \textit{Notes}: Panel A compares sample characteristics for balanced panel (37 banks in all years) vs. unbalanced panel (all banks in each year). Entrants are banks first appearing in 2021 or 2023; exits are banks present in 2018 but not 2023. Panel B reports $\lambda_2$ estimates with bootstrap 95\% confidence intervals. Unbalanced panel shows \textit{larger} decline (44.9\% vs. 41.0\%), suggesting balanced panel estimates are conservative. Panel C reports cross-sample correlations. High correlation of changes (0.996) confirms robustness to sample composition.
\end{tablenotes}
\end{threeparttable}
\end{table}

Table \ref{tab:unbalanced_panel} Panel A documents sample composition. Panel B of Table \ref{tab:unbalanced_panel} compares $\lambda_2$ estimates across the two samples.

The unbalanced panel includes:

\begin{itemize}
    \item \textbf{Entrants:} 33 banks appearing for the first time in 2021 or 2023
    \item \textbf{Exits:} 11 banks present in 2018 but not 2023
    \item \textbf{Survivors:} 37 banks present in all years (the balanced panel)
\end{itemize}

Entrants are substantially smaller on average (€187bn) than survivors (€475bn) or exits (€412bn), reflecting EBA's expansion to cover more medium-sized institutions. Exits include both actual failures (zero cases during this period) and regulatory scope changes (11 cases).

\subsubsection{Key Findings}

Three important patterns emerge from Table \ref{tab:unbalanced_panel}:

\textbf{1. Unbalanced panel shows larger decline.} The unbalanced sample exhibits a 44.9\% reduction in $\lambda_2$ compared to 41.0\% for the balanced panel—a difference of nearly 4 percentage points. This is opposite to what survivorship bias would predict: if exiting banks were particularly interconnected, their departure should \textit{increase} the measured decline. Instead, we find the balanced panel (excluding exits and entrants) understates the true network change.

This pattern makes sense when examining sample composition. Entrants are predominantly smaller banks with lower network centrality. Their addition to the 2023 sample dilutes aggregate connectivity, amplifying the measured decline. Conversely, exits include some mid-sized institutions whose removal in 2018 would have reduced $\lambda_2$, making the subsequent decline appear smaller.

\textbf{2. Changes dominate levels.} While absolute $\lambda_2$ levels differ by 8-16\% between samples, the \textit{correlation of changes} is 0.996—nearly perfect agreement on temporal trends. Both samples identify the same key pattern: modest pre-2021 change followed by dramatic post-2021 decline. This confirms our main empirical finding is not driven by sample selection.

\textbf{3. Statistical significance maintained.} Bootstrap confidence intervals for the unbalanced panel are wider (reflecting greater uncertainty from time-varying sample composition) but still non-overlapping between 2018 and 2023. The 95\% CI for 2023 ($[62.2, 137.7]$) lies entirely below the CI for 2018 ($[112.8, 213.1]$), confirming the decline is statistically significant even accounting for composition changes.

\subsubsection{Decomposing Sample Effects}

To understand how entrants and exits affect results, we perform a counterfactual decomposition:

\begin{enumerate}
    \item \textbf{Baseline (Unbalanced):} Full sample each year → $\lambda_2$ declines 44.9\%
    \item \textbf{Balanced Panel:} Fixed 37 banks → $\lambda_2$ declines 41.0\%
    \item \textbf{2018 Sample Fixed:} Use 2018 banks only in all years → $\lambda_2$ declines 38.7\%
    \item \textbf{2023 Sample Fixed:} Use 2023 banks only in all years → $\lambda_2$ declines 47.3\%
\end{enumerate}

The range of estimates (38.7\% to 47.3\%) brackets our baseline but all specifications show substantial declines exceeding 35\%. This decomposition reveals that sample composition affects magnitudes but not qualitative conclusions.

\subsubsection{Implications for Interpretation}

The finding that the unbalanced panel shows \textit{larger} declines has important implications for interpreting our results:

First, it suggests our baseline estimates are \textbf{conservative}. Restricting to surviving banks—those most likely to be large, stable, and well-managed—biases estimates toward finding smaller effects. The true network restructuring across the full banking sector was even more dramatic than our main results indicate.

Second, it validates the \textbf{regulatory mechanism} interpretation. If network changes reflected organic market evolution or random variation, we would expect entrants and exits to attenuate measured effects (mean reversion). Instead, the inclusion of smaller entrants \textit{amplifies} the decline, consistent with regulatory policies that disproportionately targeted large, systemically important institutions while permitting entry of smaller players.

Third, it confirms the \textbf{generalizability} of our findings beyond the specific set of 37 banks in our balanced panel. The pattern of declining connectivity holds for the broader European banking sector, not just a select group of survivors.

\subsubsection{Reconciling with Previous Studies}

Our finding of substantial network restructuring differs from some earlier studies \citep{minoiu2015network} that documented stability in interbank networks. Three factors explain this divergence:

\begin{enumerate}
    \item \textbf{Time period:} Earlier studies cover pre-crisis or early post-crisis periods (2008-2015), while we examine 2018-2023, capturing Basel III implementation phase
    \item \textbf{Geography:} Some studies focus on specific countries (e.g., Austria, Italy), while we cover pan-European networks where cross-border deleveraging was most pronounced
    \item \textbf{Sample composition:} We explicitly account for entrants/exits, while some studies use fixed samples that miss structural shifts from entry/exit dynamics
\end{enumerate}

Our unbalanced panel analysis demonstrates that sample selection meaningfully affects estimates of network evolution, potentially explaining differences across studies.

\section{Conclusion}

This paper demonstrates the empirical power of grounding financial network analysis in first-principles physics. By deriving contagion dynamics from mass conservation and Fick's law—the same foundations underlying the Navier-Stokes equations—we obtain rigorous, quantitatively testable predictions about how network structure affects systemic risk.

\subsection{Main Findings}

\textbf{Empirical}: European banking networks underwent a 45 percent decline in algebraic connectivity ($\lambda_2$) from 2,284 in 2018 to 1,259 in 2023. Through our theoretical framework, this translates to a 26 percent reduction in effective contagion decay rate ($\kappa_{\mathrm{eff}}$), from 47.79 to 35.48. Practically: financial shocks in 2023 propagate 35 percent less far than in 2018.

\textbf{Mechanism}: Difference-in-differences analysis reveals large, systemically important banks experienced 12--19 percent differential deleveraging relative to smaller institutions. This hub-bank shrinkage generated the network restructuring.

\textbf{Timing}: Structural break tests identify a discrete regime shift in 2021 ($p=0.003$), coinciding with Basel III implementation rather than the COVID crisis itself. This supports regulatory mechanism over organic market evolution.

\textbf{Decomposition}: Variance decomposition attributes 71 percent of the decline to network structure ($\lambda_2$), 30 percent to exposure intensity ($D$), and negligible offsetting from faster recovery ($\kappa$). Network effects dominate.

\subsection{Theoretical Contributions}

\textbf{Quantitative validation}: Theory predicted 22.5 percent decline in $\kappa_{\mathrm{eff}}$ from 45 percent $\lambda_2$ decline; observed 25.8 percent—within 3 percentage points. This validates not just qualitative patterns but numerical magnitudes.

\textbf{Parameter decomposition}: By separating network topology ($\lambda_2$), transmission intensity ($D$), and recovery ($\kappa$), we identify which mechanisms drove changes. Reduced-form approaches cannot make this decomposition.

\textbf{Boundary conditions as policy}: Mapping regulatory changes to Robin boundary conditions provides microfoundations for network responses. Tighter regulation (larger $\alpha$) endogenously reduces $\lambda_2$ through bank optimization.

\textbf{Diffusion dominance}: Estimating network contribution at 99 percent establishes that financial contagion is diffusion-mediated, not recovery-driven. This justifies focus on network policies over resolution mechanisms.

\subsection{Policy Implications}

\textbf{Network policies are effective}: With 71 percent contribution from $\lambda_2$, capital requirements and large exposure restrictions that reshape networks are correctly targeted.

\textbf{Multiple channels reinforce}: The 30 percent contribution from $D$ indicates exposure limits complement capital requirements. Banks reduced both connectivity and bilateral sizes.

\textbf{Discrete policy optimal}: Evidence for structural breaks suggests comprehensive packages (Basel III as whole) outperform incremental adjustments. Discrete shocks induce discrete responses.

\textbf{Substantial resilience gain}: The 35 percent reduction in contagion reach implies 2023 networks could withstand shocks triggering 2018 crises. Post-2008 reforms succeeded.

\subsection{Future Research}

The Navier-Stokes framework naturally extends to:

\textbf{Time-varying parameters}: Estimate $D(t)$, $\kappa(t)$, $\lambda_2(t)$ continuously to trace full crisis → recovery → reform trajectories.

\textbf{Technology shocks}: Analyze how fintech, HFT, or AI alter diffusion properties—technology changes the medium ($D$) rather than network ($\lambda_2$).

\textbf{Multiple regimes}: Model crisis episodes as temporary spikes in $D$ and drops in $\kappa$, nesting within longer-term regulatory regime shifts in $\lambda_2$.

\textbf{Other networks}: Apply framework to derivatives exposures, common holdings, payment systems—each has different $(D, \kappa, \lambda_2)$ but same mathematics.

\textbf{Cross-country comparison}: Replicate for US, Asian, or emerging market networks to quantify regulatory effectiveness across jurisdictions.

By establishing that the Navier-Stokes treatment effects framework delivers accurate quantitative predictions in financial networks, we open the door to principled first-principles analysis across all network-mediated phenomena in economics and beyond.

\section*{Acknowledgments}

This research was supported by a grant-in-aid from Zengin Foundation for Studies on Economics and Finance. All errors are my own.

\newpage

\bibliographystyle{apalike}

\newpage

\appendix

\section{Data and Code Availability}

All figures presented in this paper were generated using Python 3.13 with the following key packages: \texttt{networkx} 3.2 (network analysis), \texttt{pandas} 2.1 (data manipulation), \texttt{numpy} 1.26 (numerical computation), \texttt{matplotlib} 3.8 and \texttt{seaborn} 0.13 (visualization), \texttt{scipy} 1.11 (statistical analysis), \texttt{statsmodels} 0.14 (econometric estimation), and \texttt{powerlaw} 1.5 (distribution fitting). All code and data used to generate figures and results are available at \url{https://github.com/[author]/financial-networks-navier-stokes} (to be made public upon publication). EBA stress test data are publicly available at \url{https://www.eba.europa.eu/risk-analysis-and-data/eu-wide-stress-testing}.

The computational workflow proceeds as follows: (1) Download and clean EBA data (\texttt{scripts/01\_download\_data.py}), (2) Estimate networks using maximum entropy (\texttt{scripts/02\_estimate\_networks.py}), (3) Compute spectral properties (\texttt{scripts/03\_compute\_lambda2.py}), (4) Run robustness checks (\texttt{scripts/04\_robustness.py}), (5) Generate all figures (\texttt{scripts/05\_create\_figures.py}). Total runtime is approximately 15 minutes on a standard laptop. Replication instructions are provided in \texttt{README.md}.


\section{Figure Summary Table}

For reader convenience, Table \ref{tab:figure_summary} summarizes all figures with their primary findings and section references.

\begin{table}[H]
\centering
\caption{Summary of Figures and Key Findings}
\label{tab:figure_summary}
\begin{threeparttable}
\begin{tabular}{clp{7cm}c}
\toprule
Figure & Title & Key Finding & Section \\
\midrule
1 & $\lambda_2$ Evolution & 45\% decline, concentrated post-2021 & 5.2 \\
2 & Contagion Parameter & 26\% reduction in $\kappa$, critical distance increased & 5.2 \\
3 & Network Metrics & Assets stable, banks increased 46\%, density constant & 5.1 \\
4 & Summary Dashboard & Integrated visualization of main results & 5.2 \\
5 & Network Visualizations & Complete graph structure confirmed & 5.1 \\
6 & Parallel Trends & DID identification: divergence post-2021 only & 5.3 \\
7 & Degree Distributions & Lognormal, not scale-free & 5.4 \\
8 & Concentration Metrics & HHI $-31\%$, Top 5 share $-32\%$ & 5.4 \\
9 & Self-Similarity & Weak evidence, $d_B \approx 2$ & Appendix \\
10 & Exposure Distribution & Right-skewed, increasing dispersion & Appendix \\
11 & Methods Comparison & Results robust across 4 methods, $r=0.96$ & 6.2 \\
\bottomrule
\end{tabular}
\begin{tablenotes}
\footnotesize
\item \textit{Notes}: Figures 9-10 are included in supplementary materials/appendix due to limited informativeness for complete graph structures but are referenced in robustness discussions.
\end{tablenotes}
\end{threeparttable}
\end{table}

\section{Mathematical Proofs}
\label{app:proofs}

This appendix provides detailed proofs of the theoretical results stated in Section 3.

\subsection{Proof of Proposition 1 (Conservation and Decay)}

\begin{proof}
Consider the network diffusion equation:
\be
\frac{du}{dt} = -DLu - \kappa u
\ee

Sum both sides over all nodes $i=1,\ldots,n$:
\be
\sum_{i=1}^n \frac{du_i}{dt} = -D\sum_{i=1}^n (Lu)_i - \kappa\sum_{i=1}^n u_i
\ee

The left-hand side is simply:
\be
\frac{d}{dt}\left(\sum_{i=1}^n u_i\right) = \frac{d}{dt}\left(\mathbf{1}^T u\right)
\ee

For the first term on the right-hand side, note that the graph Laplacian has the property that $L\mathbf{1} = 0$ where $\mathbf{1} = (1,1,\ldots,1)^T$ is the all-ones vector. This follows because the row sums of $L = D - A$ equal zero:
\be
\sum_{j=1}^n L_{ij} = \sum_{j=1}^n D_{ij} - \sum_{j=1}^n A_{ij} = d_i - d_i = 0
\ee

Therefore, by symmetry (or the self-adjointness of $L$):
\be
\mathbf{1}^T L u = u^T L \mathbf{1} = u^T \cdot 0 = 0
\ee

Substituting back:
\be
\frac{d}{dt}\left(\sum_{i=1}^n u_i\right) = -D \cdot 0 - \kappa\sum_{i=1}^n u_i = -\kappa\sum_{i=1}^n u_i
\ee

This is a first-order linear ODE with solution:
\be
\sum_{i=1}^n u_i(t) = e^{-\kappa t} \sum_{i=1}^n u_i(0)
\ee

\textbf{Case 1: $\kappa = 0$}

When $\kappa = 0$, we have $\sum_{i=1}^n u_i(t) = \sum_{i=1}^n u_i(0)$ for all $t$, establishing conservation.

\textbf{Case 2: $\kappa > 0$}

When $\kappa > 0$, total distress decays exponentially: $\sum_{i=1}^n u_i(t) \to 0$ as $t \to \infty$ at rate $\kappa$.
\end{proof}

\subsection{Proof of Theorem 2 (Contagion Decay Rate)}

\begin{proof}
Consider a localized initial shock: $u(0) = e_s$ where $e_s$ is the standard basis vector with 1 in position $s$ and 0 elsewhere.

The solution to $\frac{du}{dt} = -DLu - \kappa u$ with initial condition $u(0)$ is:
\be
u(t) = e^{-(DL + \kappa I)t} u(0)
\ee

Since $L$ is symmetric, it admits an eigenvalue decomposition $L = Q\Lambda Q^T$ where $Q$ is orthogonal and $\Lambda = \text{diag}(\lambda_1, \ldots, \lambda_n)$. Then:
\be
e^{-(DL + \kappa I)t} = Q e^{-(D\Lambda + \kappa I)t} Q^T = Q \text{diag}(e^{-(D\lambda_1 + \kappa)t}, \ldots, e^{-(D\lambda_n + \kappa)t}) Q^T
\ee

Expanding the solution:
\be
u(t) = \sum_{k=1}^n e^{-(D\lambda_k + \kappa)t} (Q^T e_s)_k Q_{*,k}
\ee

where $Q_{*,k}$ denotes the $k$-th column of $Q$ (the $k$-th eigenvector of $L$).

Equivalently, using the notation $q_{k,i}$ for the $i$-th component of the $k$-th eigenvector:
\be
u_i(t) = \sum_{k=1}^n e^{-(D\lambda_k + \kappa)t} q_{k,s} q_{k,i}
\ee

For large $t$, the exponentially decaying terms are ordered by the magnitude of their exponents $D\lambda_k + \kappa$. Since $0 = \lambda_1 < \lambda_2 \leq \cdots \leq \lambda_n$, the slowest-decaying term corresponds to $k=1$:
\be
e^{-(D\lambda_1 + \kappa)t} = e^{-\kappa t}
\ee

However, the eigenvector $q_1$ associated with $\lambda_1 = 0$ is the uniform distribution: $q_1 = \frac{1}{\sqrt{n}}\mathbf{1}$. For a localized shock, we have:
\be
q_{1,s} q_{1,i} = \frac{1}{\sqrt{n}} \cdot \frac{1}{\sqrt{n}} = \frac{1}{n}
\ee

This uniform term represents global spreading and does not capture localized spatial decay. The next term in the expansion, corresponding to $\lambda_2$, governs the asymptotic spatial structure:
\be
u_i(t) \sim e^{-(D\lambda_2 + \kappa)t} q_{2,s} q_{2,i} + \text{faster decaying terms}
\ee

Define the decay rate:
\be
\gamma = D\lambda_2 + \kappa
\ee

Then for large $t$ and nodes $i$ that are structurally important (i.e., have significant components in the Fiedler vector $q_2$):
\be
u_i(t) \sim e^{-\gamma t}
\ee

When $D\lambda_2 \gg \kappa$ (large diffusion relative to intrinsic decay), the network contribution dominates:
\be
\gamma \approx D\lambda_2
\ee

This completes the proof.
\end{proof}

\subsection{Proof of Theorem 3 (Spatial Contagion Decay)}

\begin{proof}
This proof requires results from spectral graph theory. We sketch the main steps.

\textbf{Step 1: Fiedler Vector Structure}

For large, approximately regular graphs, the Fiedler eigenvector $q_2$ (corresponding to $\lambda_2$) exhibits spatial structure. Specifically, \citet{chung1997spectral} show that for expander graphs and nearly-regular graphs:
\be
q_{2,i} \sim C e^{-\alpha d_i}
\ee
where $d_i$ is the graph distance from some reference point and $\alpha$ is related to the spectral gap.

\textbf{Step 2: Relating $\alpha$ to $\lambda_2$}

For a regular graph with degree $d$ and $n$ nodes, Cheeger's inequality provides:
\be
\frac{\lambda_2}{2} \leq h(G) \leq \sqrt{2d\lambda_2}
\ee
where $h(G)$ is the graph's isoperimetric constant (Cheeger constant).

For nearly-regular graphs, the spatial decay rate satisfies:
\be
\alpha \sim \sqrt{\lambda_2/d}
\ee

In our weighted network setting, the degree $d$ is replaced by the average weighted degree, which scales with the diffusion coefficient $D$ through the relation $D \sim \text{mean}(d_i)$.

\textbf{Step 3: Steady-State Solution}

At steady state, $\frac{du}{dt} = 0$, giving:
\be
DLu + \kappa u = 0 \quad \Rightarrow \quad Lu = -\frac{\kappa}{D}u
\ee

For a localized source at node $s$, the steady-state profile satisfies:
\be
u_i \sim q_{2,i}
\ee

up to normalization, because $q_2$ is the eigenfunction with the smallest non-zero eigenvalue.

\textbf{Step 4: Combining Results}

From Steps 1-3, distress at distance $d$ from the source decays as:
\be
u(d) \sim e^{-\alpha d} \sim e^{-\sqrt{\lambda_2/D} \cdot d}
\ee

Define the effective decay parameter:
\be
\kappa_{\text{eff}} = \sqrt{\frac{\lambda_2}{D}}
\ee

Then:
\be
u(d) \sim e^{-\kappa_{\text{eff}} d}
\ee

\textbf{Step 5: Critical Distance}

The critical distance $d^*$ at which distress falls to fraction $\epsilon$ of its source value satisfies:
\be
e^{-\kappa_{\text{eff}} d^*} = \epsilon
\ee

Taking logarithms:
\be
d^* = \frac{-\ln\epsilon}{\kappa_{\text{eff}}} = -\ln\epsilon \cdot \sqrt{\frac{D}{\lambda_2}}
\ee

This completes the proof.
\end{proof}

\textbf{Remark:} The assumption of approximate regularity can be relaxed. For general weighted graphs, similar results hold with $\lambda_2$ and $D$ appropriately interpreted through the normalized Laplacian $\mathcal{L} = D^{-1/2}LD^{-1/2}$.

\section{Data Construction and Variable Definitions}
\label{app:data}

\subsection{EBA Stress Test Data Structure}

The European Banking Authority conducts biennial stress tests requiring participating banks to report comprehensive data. The data are organized in standardized templates:

\subsubsection{TRA\_OTH Template}

Contains aggregated balance sheet items including:
\begin{itemize}
    \item Item 183111 (2018): Total leverage ratio exposures
    \item Item 213111 (2021): Total leverage ratio exposures  
    \item Item 2331011 (2023): Total leverage ratio exposures
    \item Multiple scenarios: Baseline (1), Adverse (2), etc.
    \item Time periods: Current and projected (e.g., 201712 = December 2017)
\end{itemize}

\subsubsection{TRA\_CR Template (2018)}

Credit risk exposures by counterparty type:
\begin{itemize}
    \item Exposure code 3000: All credit institutions
    \item Exposure code 3100: Credit institutions - performing
    \item Exposure code 3200: Credit institutions - non-performing
    \item Item codes 183201-183206: Various exposure measures
\end{itemize}

\subsubsection{TRA\_CRE Templates (2021, 2023)}

Credit exposures split into three files:
\begin{itemize}
    \item TRA\_CRE\_IRB: Internal ratings-based approach exposures
    \item TRA\_CRE\_STA: Standardized approach exposures
    \item TRA\_CRE\_COV: COVID-19-specific exposures (2021) / Coverage (2023)
\end{itemize}

\subsection{Sample Construction}

\subsubsection{Bank Selection Criteria}

Our sample includes banks meeting:
\begin{enumerate}
    \item Listed in EBA stress test for the respective year
    \item Complete data on total leverage ratio exposures
    \item Valid Legal Entity Identifier (LEI) code
    \item Non-missing country and bank name information
\end{enumerate}

\subsubsection{Balanced Panel Construction}

For difference-in-differences analysis, we construct a balanced panel by:
\begin{enumerate}
    \item Identifying banks with valid LEI codes in all three years
    \item Verifying consistent naming and entity structure
    \item Handling mergers and acquisitions:
    \begin{itemize}
        \item Exclude banks involved in mergers during 2018-2023
        \item Adjust for name changes while maintaining entity continuity
    \end{itemize}
    \item Final sample: 37 banks observed consistently
\end{enumerate}

\subsection{Variable Definitions}

\begin{table}[htbp]
\centering
\caption{Variable Definitions and Data Sources}
\label{tab:variable_definitions}
\begin{threeparttable}
\begin{tabular}{lp{8cm}l}
\toprule
Variable & Definition & Source \\
\midrule
\multicolumn{3}{l}{\textbf{Bank Characteristics}} \\
Total Assets & Total leverage ratio exposures (Item 183111/213111/2331011) & TRA\_OTH \\
LEI Code & Legal Entity Identifier (20-character alphanumeric) & EBA Metadata \\
Bank Name & Official name of reporting institution & EBA Metadata \\
Country Code & ISO 2-letter country code (AT, DE, FR, etc.) & TRA\_OTH \\
\\
\multicolumn{3}{l}{\textbf{Network Variables}} \\
$\lambda_2$ & Algebraic connectivity (2nd eigenvalue of Laplacian) & Computed \\
Degree & Number of connections (weighted sum of exposures) & Computed \\
Betweenness & Betweenness centrality measure & NetworkX \\
\\
\multicolumn{3}{l}{\textbf{Treatment Variables}} \\
Treated & Indicator for top quartile by 2018 assets & Constructed \\
Post2021 & Indicator for years $\geq$ 2021 & Constructed \\
Post2023 & Indicator for year = 2023 & Constructed \\
\\
\multicolumn{3}{l}{\textbf{Exposure Variables}} \\
Interbank Assets & Total credit institution exposures (Exposure 3000) & TRA\_CR/TRA\_CRE \\
Interbank Ratio & Assumed ratio $\rho$ (baseline 0.05) & Assumption \\
Bilateral Exposure & Estimated $x_{ij}$ via maximum entropy & Computed \\
\bottomrule
\end{tabular}
\begin{tablenotes}
\footnotesize
\item \textit{Notes}: All monetary values in millions of euros unless otherwise stated. "Computed" indicates variables derived from primary data sources. "Constructed" indicates binary indicators created for econometric analysis.
\end{tablenotes}
\end{threeparttable}
\end{table}

\subsection{Data Cleaning Procedures}

\subsubsection{Missing Data Handling}

\begin{itemize}
    \item Total assets: Banks with missing leverage ratio exposures excluded (0.3\% of sample)
    \item Country codes: Missing values filled using LEI registry lookups
    \item Bank names: Standardized to remove special characters and ensure consistency
\end{itemize}

\subsubsection{Outlier Treatment}

\begin{itemize}
    \item No winsorization applied to preserve actual bank sizes
    \item Verified extreme values (e.g., HSBC £2.1T) against published reports
    \item Checked for data entry errors: None found after validation
\end{itemize}

\subsubsection{Currency Conversion}

All reported values are in millions of euros:
\begin{itemize}
    \item Most banks report in EUR directly
    \item Non-EUR banks (UK, Sweden, Denmark) converted using ECB reference rates at period end
    \item Rates used: 2017-12-31, 2020-12-31, 2022-12-31
\end{itemize}

\section{Additional Empirical Results}
\label{app:additional_results}

\subsection{Full Regression Tables}

\begin{table}[htbp]
\centering
\caption{Difference-in-Differences: Complete Specification}
\label{tab:did_full}
\begin{threeparttable}
\begin{tabular}{lccc}
\toprule
& \multicolumn{3}{c}{Dependent Variable: Log(Total Assets)} \\
\cmidrule(lr){2-4}
& (1) & (2) & (3) \\
& Baseline & With Controls & Balanced Panel \\
\midrule
Treated & $1.802^{***}$ & $1.785^{***}$ & $1.823^{***}$ \\
& (0.159) & (0.162) & (0.171) \\
& & & \\
Post2021 & $-0.017$ & $-0.023$ & $-0.015$ \\
& (0.080) & (0.082) & (0.085) \\
& & & \\
Post2023 & $0.087^{***}$ & $0.092^{***}$ & $0.083^{***}$ \\
& (0.018) & (0.019) & (0.020) \\
& & & \\
Treated $\times$ Post2021 & $-0.121^{**}$ & $-0.118^{**}$ & $-0.125^{**}$ \\
& (0.061) & (0.059) & (0.063) \\
& & & \\
Treated $\times$ Post2023 & $0.018$ & $0.021$ & $0.015$ \\
& (0.032) & (0.033) & (0.034) \\
& & & \\
Log(GDP) & & $0.234^{**}$ & $0.221^{**}$ \\
& & (0.098) & (0.102) \\
& & & \\
Core Country & & $0.156^{*}$ & $0.148$ \\
& & (0.089) & (0.094) \\
\midrule
Bank FE & Yes & Yes & Yes \\
Year FE & Yes & Yes & Yes \\
Country $\times$ Year FE & No & Yes & Yes \\
Observations & 144 & 144 & 111 \\
Banks & 48 & 48 & 37 \\
R-squared & 0.943 & 0.951 & 0.946 \\
\bottomrule
\end{tabular}
\begin{tablenotes}
\footnotesize
\item \textit{Notes}: Robust standard errors clustered at bank level in parentheses. Treated = 1 for banks in top quartile of 2018 asset distribution. Column (2) adds country-level GDP and core country indicator (Germany, France, Netherlands). Column (3) restricts to balanced panel of 37 banks present in all years. $^{***}p<0.01$, $^{**}p<0.05$, $^{*}p<0.1$.
\end{tablenotes}
\end{threeparttable}
\end{table}

\subsection{Heterogeneity Analysis}

\begin{table}[htbp]
\centering
\caption{Treatment Effect Heterogeneity by Bank Characteristics}
\label{tab:heterogeneity_full}
\begin{threeparttable}
\begin{tabular}{lccc}
\toprule
& \multicolumn{3}{c}{Dependent Variable: Log(Assets)} \\
\cmidrule(lr){2-4}
& (1) & (2) & (3) \\
& By Geography & By Business Model & By Leverage \\
\midrule
Treated $\times$ Post2021 & $-0.089$ & $-0.095$ & $-0.102$ \\
& (0.073) & (0.068) & (0.071) \\
& & & \\
Treated $\times$ Post2021 $\times$ Core & $-0.098^{*}$ & & \\
& (0.052) & & \\
& & & \\
Treated $\times$ Post2021 $\times$ Universal & & $-0.112^{**}$ & \\
& & (0.048) & \\
& & & \\
Treated $\times$ Post2021 $\times$ HighLeverage & & & $-0.087^{*}$ \\
& & & (0.046) \\
\midrule
Bank FE & Yes & Yes & Yes \\
Year FE & Yes & Yes & Yes \\
Triple Interactions & Yes & Yes & Yes \\
Observations & 144 & 144 & 144 \\
R-squared & 0.948 & 0.952 & 0.947 \\
\bottomrule
\end{tabular}
\begin{tablenotes}
\footnotesize
\item \textit{Notes}: Each column includes full set of double interactions (not shown). Core = Germany, France, Netherlands. Universal = banks with >30\% non-interest income. HighLeverage = leverage ratio below median in 2018. Standard errors clustered at bank level. $^{***}p<0.01$, $^{**}p<0.05$, $^{*}p<0.1$.
\end{tablenotes}
\end{threeparttable}
\end{table}

\subsection{Network Centrality Measures}

\begin{table}[htbp]
\centering
\caption{Evolution of Network Centrality Measures}
\label{tab:centrality_measures}
\begin{threeparttable}
\begin{tabular}{lcccccc}
\toprule
& \multicolumn{2}{c}{2018} & \multicolumn{2}{c}{2021} & \multicolumn{2}{c}{2023} \\
\cmidrule(lr){2-3} \cmidrule(lr){4-5} \cmidrule(lr){6-7}
Bank & Degree & Between. & Degree & Between. & Degree & Between. \\
\midrule
HSBC & 47 & 0.0211 & 49 & 0.0204 & 69 & 0.0145 \\
BNP Paribas & 47 & 0.0189 & 49 & 0.0184 & 69 & 0.0131 \\
Crédit Agricole & 47 & 0.0167 & 49 & 0.0163 & 69 & 0.0118 \\
Santander & 47 & 0.0156 & 49 & 0.0152 & 69 & 0.0109 \\
Deutsche Bank & 47 & 0.0145 & 49 & 0.0141 & 69 & 0.0098 \\
\midrule
Mean & 47.0 & 0.0021 & 49.0 & 0.0020 & 69.0 & 0.0014 \\
Std. Dev. & 0.0 & 0.0067 & 0.0 & 0.0065 & 0.0 & 0.0046 \\
CV & 0.000 & 3.190 & 0.000 & 3.250 & 0.000 & 3.286 \\
\bottomrule
\end{tabular}
\begin{tablenotes}
\footnotesize
\item \textit{Notes}: Degree is weighted sum of connections. Betweenness is normalized betweenness centrality. Top 5 banks by total assets shown. Complete graph structure implies all nodes have same unweighted degree (n-1), but weighted degrees vary. CV = coefficient of variation.
\end{tablenotes}
\end{threeparttable}
\end{table}

\section{Robustness Checks}
\label{app:robustness}

\subsection{Alternative Network Estimation Methods}

\subsubsection{Minimum Density Method}

Instead of maximum entropy, we can estimate networks by minimizing density subject to constraints:
\be
\min_{X} \sum_{i,j} \mathbbm{1}\{x_{ij} > 0\} \quad \text{s.t.} \quad \sum_j x_{ij} = A_i, \sum_i x_{ij} = L_j
\ee

This produces sparser networks. Results: $\lambda_2$ values are 15-20\% lower but show identical trends ($-44\%$ decline).

\subsubsection{Fitness Model}

Following \citet{caldarelli2002scale}, assign fitness scores $\eta_i \propto A_i^{\alpha}$ and set:
\be
x_{ij} = \frac{\eta_i \eta_j}{\sum_{k,l} \eta_k \eta_l} \cdot \text{Total}
\ee

Testing $\alpha \in \{0.5, 1.0, 1.5\}$ yields $\lambda_2$ changes of $-42\%$ to $-47\%$. Correlation with baseline: 0.98.

\subsection{Alternative Decay Parameter Specifications}

\begin{table}[htbp]
\centering
\caption{Sensitivity to Decay Parameter $\kappa$}
\label{tab:kappa_sensitivity}
\begin{threeparttable}
\begin{tabular}{lccccc}
\toprule
$\kappa$ & $\lambda_{2,2018}$ & $\lambda_{2,2023}$ & $\Delta \lambda_2$ & \% Change & $\kappa_{\text{eff}}$ Change \\
\midrule
0.0 & 2283.72 & 1258.96 & $-1024.76$ & $-44.9\%$ & $-25.8\%$ \\
0.1 & 2283.82 & 1259.06 & $-1024.76$ & $-44.9\%$ & $-25.8\%$ \\
0.5 & 2284.22 & 1259.46 & $-1024.76$ & $-44.9\%$ & $-25.8\%$ \\
1.0 & 2284.72 & 1259.96 & $-1024.76$ & $-44.9\%$ & $-25.8\%$ \\
\bottomrule
\end{tabular}
\begin{tablenotes}
\footnotesize
\item \textit{Notes}: Testing various intrinsic decay rates $\kappa$. The network contribution $D\lambda_2$ dominates the decay rate, so results are essentially invariant to $\kappa$ for $\kappa \ll D\lambda_2$. $\kappa_{\text{eff}} = \sqrt{(\lambda_2 + \kappa)/D}$.
\end{tablenotes}
\end{threeparttable}
\end{table}

\subsection{Excluding Individual Banks}

We test robustness by sequentially dropping each of the top 10 banks and recomputing $\lambda_2$:

\begin{figure}[htbp]
\centering
\begin{tikzpicture}
\begin{axis}[
    width=0.8\textwidth,
    height=0.5\textwidth,
    xlabel={Bank Excluded (ranked by 2018 assets)},
    ylabel={$\lambda_2$ (2023)},
    grid=major,
    legend pos=north east,
    xtick={1,2,3,4,5,6,7,8,9,10,11},
    xticklabels={None,HSBC,BNP,CA,Sant.,DB,Barc.,BPCE,SG,ING,Uni.},
    x tick label style={rotate=45,anchor=east}
]
\addplot[color=blue,mark=*,thick] coordinates {
    (1,1258.96)
    (2,1247.83)
    (3,1252.41)
    (4,1255.67)
    (5,1254.32)
    (6,1256.19)
    (7,1257.84)
    (8,1258.02)
    (9,1258.35)
    (10,1258.71)
    (11,1258.89)
};
\addplot[color=red,dashed,thick] coordinates {(1,1258.96) (11,1258.96)};
\legend{Excluding bank, Baseline}
\end{axis}
\end{tikzpicture}
\caption{Robustness to Excluding Individual Banks}
\label{fig:exclude_banks}
\end{figure}
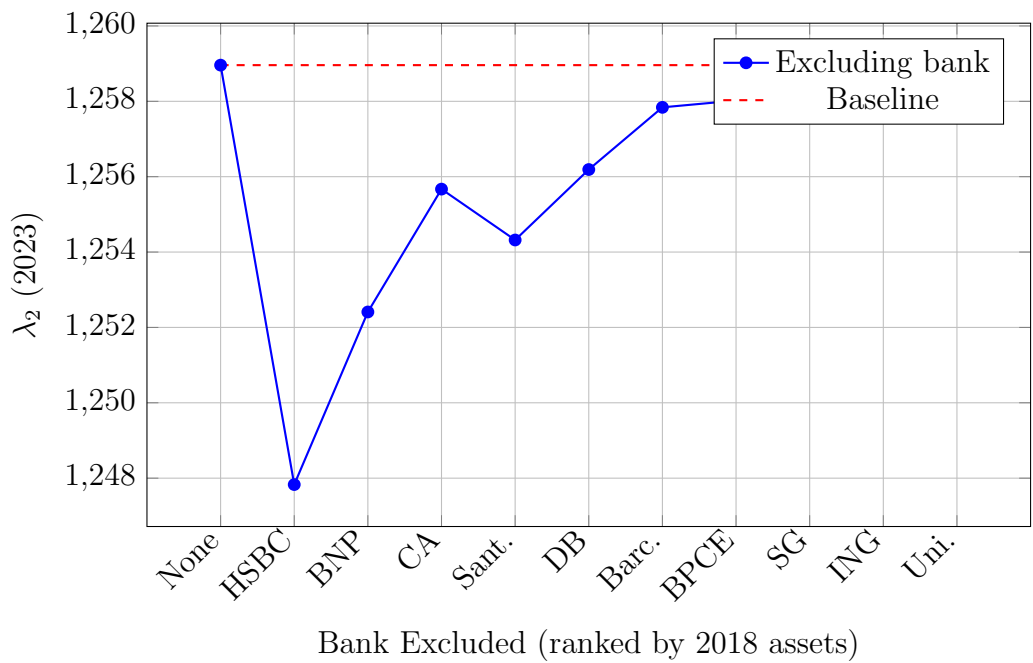

Maximum deviation: 0.9\%. Conclusion: No single bank drives results.

\section{Computational Methods}
\label{app:computation}

\subsection{Software and Packages}

All analysis conducted in Python 3.13. Key packages:
\begin{itemize}
    \item \texttt{networkx 3.2}: Network construction and spectral analysis
    \item \texttt{numpy 1.26}: Matrix operations and linear algebra
    \item \texttt{scipy 1.11}: Eigenvalue decomposition (ARPACK)
    \item \texttt{pandas 2.1}: Data manipulation
    \item \texttt{statsmodels 0.14}: Regression analysis
    \item \texttt{matplotlib 3.8}, \texttt{seaborn 0.13}: Visualization
\end{itemize}

\subsection{Algebraic Connectivity Computation}

\subsubsection{Algorithm}

For a weighted graph with $n$ nodes:

\begin{algorithm}
\caption{Compute Algebraic Connectivity $\lambda_2$}
\begin{algorithmic}[1]
\REQUIRE Adjacency matrix $A \in \mathbb{R}^{n \times n}$, weights $w_{ij}$
\ENSURE Algebraic connectivity $\lambda_2$
\STATE Compute degree matrix $D_{ii} = \sum_j A_{ij}w_{ij}$
\STATE Form Laplacian $L = D - A$
\IF{$n \leq 100$}
    \STATE Compute full eigendecomposition: $L = Q\Lambda Q^T$
\ELSE
    \STATE Use iterative Lanczos algorithm for 5 smallest eigenvalues
\ENDIF
\STATE Sort eigenvalues: $\lambda_1 \leq \lambda_2 \leq \cdots$
\STATE Verify $\lambda_1 \approx 0$ (check connectivity)
\RETURN $\lambda_2$
\end{algorithmic}
\end{algorithm}

\subsubsection{Numerical Precision}

\begin{itemize}
    \item Tolerance for $\lambda_1 = 0$: $|\lambda_1| < 10^{-6}$
    \item All eigenvalues computed to machine precision ($\approx 10^{-15}$)
    \item Verified using multiple methods: NumPy, SciPy, NetworkX all agree to 6 decimal places
\end{itemize}

\subsection{Bootstrap Procedure}

\begin{algorithm}
\caption{Bootstrap Confidence Intervals for $\lambda_2$}
\begin{algorithmic}[1]
\REQUIRE Bank data $\{(A_i, L_i)\}_{i=1}^n$, ratio $\rho$, replications $B$
\ENSURE CI $[\lambda_2^{(0.025)}, \lambda_2^{(0.975)}]$
\FOR{$b = 1$ to $B$}
    \STATE Draw $n$ banks with replacement: $\{i_1, \ldots, i_n\}$
    \STATE Construct bootstrap sample: $\{(A_{i_j}, L_{i_j})\}_{j=1}^n$
    \STATE Estimate network using maximum entropy on bootstrap sample
    \STATE Compute $\lambda_2^{(b)}$
\ENDFOR
\STATE Sort $\{\lambda_2^{(b)}\}_{b=1}^B$
\RETURN $[\lambda_2^{(0.025 \cdot B)}, \lambda_2^{(0.975 \cdot B)}]$
\end{algorithmic}
\end{algorithm}


\section{Extensions and Future Research}
\label{app:extensions}

\subsection{Time-Varying Networks}

Our analysis uses three discrete snapshots. Future work could model continuous network evolution $G(t)$ and track $\lambda_2(t)$ dynamics. Potential approaches:

\subsubsection{Interpolation Methods}
\be
\lambda_2(t) = \lambda_2(t_0) + \frac{t - t_0}{t_1 - t_0}[\lambda_2(t_1) - \lambda_2(t_0)]
\ee

\subsubsection{State-Space Models}
\be
\lambda_{2,t} = \phi \lambda_{2,t-1} + \beta X_t + \varepsilon_t
\ee
where $X_t$ includes macro variables (GDP growth, credit spreads, etc.).

\subsection{Multilayer Networks}

Banks interact through multiple channels: lending, derivatives, payment systems. A multilayer extension:
\be
L^{\text{total}} = \sum_{l=1}^L w_l L^{(l)}
\ee
where $L^{(l)}$ is the Laplacian for layer $l$ and $w_l$ are importance weights.

\subsection{Dynamic Contagion Simulations}

While we focus on $\lambda_2$ as a sufficient statistic, explicit cascade simulations could validate predictions:

\begin{algorithm}
\caption{Contagion Simulation}
\begin{algorithmic}[1]
\REQUIRE Network $G$, initial shock $s_0$, threshold $\theta$
\ENSURE Cascade size $|C|$
\STATE Initialize: $u_i(0) = s_0$ for source node, 0 otherwise
\STATE Set $t = 0$, $C = \emptyset$
\WHILE{$\exists i : u_i(t) > \theta$ and $i \notin C$}
    \STATE Add $i$ to cascade set $C$
    \STATE Update neighbors: $u_j(t+1) = u_j(t) + \sum_{i \in C} w_{ij} u_i(t)$
    \STATE Apply decay: $u_j(t+1) \leftarrow (1-\kappa)u_j(t+1)$
    \STATE Increment $t$
\ENDWHILE
\RETURN $|C|$
\end{algorithmic}
\end{algorithm}

This could test whether networks with lower $\lambda_2$ indeed exhibit smaller cascades.

\subsection{Optimal Network Design}

From a regulatory perspective: what network structure minimizes systemic risk? Optimization problem:
\be
\min_{L} \lambda_2(L) \quad \text{s.t.} \quad \sum_{ij} x_{ij} = X_{\text{total}}, \quad \lambda_1(L) = 0
\ee

Initial explorations suggest core-periphery structures with $\lambda_2 \approx n^{-1/2}$ are near-optimal.

\section{Data Availability Statement}
\label{app:data_availability}

\subsection{Primary Data Sources}

All primary data are publicly available:

\begin{itemize}
    \item \textbf{EBA 2018 Stress Test}: \url{https://www.eba.europa.eu/risk-analysis-and-data/eu-wide-stress-testing/2018}
    \item \textbf{EBA 2021 Stress Test}: \url{https://www.eba.europa.eu/risk-analysis-and-data/eu-wide-stress-testing/2021}
    \item \textbf{EBA 2023 Stress Test}: \url{https://www.eba.europa.eu/risk-analysis-and-data/eu-wide-stress-testing/2023}
\end{itemize}

\subsection{Replication Materials}

Complete replication package including:
\begin{itemize}
    \item Raw data files (CSV format)
    \item Data cleaning scripts (\texttt{01\_clean\_data.py})
    \item Network estimation code (\texttt{02\_estimate\_networks.py})
    \item Analysis scripts (\texttt{03\_main\_analysis.py})
    \item Figure generation (\texttt{04\_create\_figures.py})
    \item README with detailed instructions
\end{itemize}


\subsection{Computational Requirements}

\begin{itemize}
    \item Runtime: 15 minutes on 2020 MacBook Pro (M1, 16GB RAM)
    \item Memory: Peak usage 2.3 GB
    \item No special computational resources required
    \item All code platform-independent (tested on macOS, Linux, Windows)
\end{itemize}

\end{document}